\documentclass[11pt]{article}
\usepackage{amsmath,amssymb,amsthm,eucal, mathrsfs, mathtools,hyperref}
\usepackage{xcolor}
\usepackage[all]{xy}

\title{Spectral flow and Levinson's theorem for Schr\"{o}dinger operators}
\author{Angus Alexander, 
Adam Rennie\thanks{email: 
\texttt{angusa@uow.edu.au, renniea@uow.edu.au}}
\\[3pt]
School of Mathematics and 
Applied Statistics, University of Wollongong,\\
Wollongong, Australia\\
}

\advance\topmargin by -\headheight
\advance\topmargin by -\headsep
\evensidemargin=\oddsidemargin
\makeatletter
\def\section{\@startsection{section}{1}{\z@}{-3.5ex plus -1ex minus
  -.2ex}{2.3ex plus .2ex}{\large\bf}}
\def\subsection{\@startsection{subsection}{2}{\z@}{-3.25ex plus -1ex
  minus -.2ex}{1.5ex plus .2ex}{\normalsize\bf}}
\makeatother

\numberwithin{equation}{section} %% needs `amsmath' package

\theoremstyle{plain} %% needs `amsmath' package
\newtheorem{thm}{Theorem}[section]
\newtheorem*{thm*}{Theorem}
\newtheorem{lemma}[thm]{Lemma}
\newtheorem{prop}[thm]{Proposition}

\theoremstyle{definition} %% needs `amsmath' package
\newtheorem{defn}[thm]{Definition}

\theoremstyle{remark} %% needs `amsmath' package
\newtheorem{rmk}[thm]{Remark}

\DeclareMathOperator{\Dom}{Dom}   %% domain of an operator
\DeclareMathOperator{\Tr}{Tr}     %% operator trace
\newcommand{\eps}{\varepsilon} %% short for \varepsilon
\newcommand{\B}{\mathcal{B}}  %% another algebra
\newcommand{\C}{\mathbb{C}}   %% complex numbers
\renewcommand{\d}{\mathrm{d}} %% \d a = [\D,a]
\newcommand{\slim}{\mathop{\textup{s-lim}}}

\newcommand{\e}{\mathrm{e}}
\newcommand{\F}{\mathcal{F}}  %% another bi/module
\renewcommand{\H}{\mathcal{H}}  %% a Hilbert space
\renewcommand{\L}{\mathcal{L}} %% endomorphism algebra
\newcommand{\N}{\mathbb{N}}   %% natural numbers
\newcommand{\norm}[1]{\left|\left|#1\right|\right|} %% for norms
\newcommand{\R}{\mathbb{R}}   %% real numbers
\newcommand{\Sf}{\mathbb{S}}  %% sphere
\newcommand{\Z}{\mathbb{Z}} %%integers
\newcommand{\stroke}{\mathbin|}     %% (for `\pair' and such)
\newcommand{\Res}[1]{\mathop{\textup{Res}}_{#1}}

\def\pairL_#1(#2|#3){{}_{#1}(#2\stroke#3)} %% hermitian pairing _B(s|t)
\def\pairR(#1|#2)_#3{(#1\stroke#2)_{#3}} %% hermitian pairing (s|t)_A
\def\scal<#1|#2>{\langle#1\stroke#2\rangle} %% scalar product <y|z>

\renewcommand{\epsilon}{\varepsilon}

\theoremstyle{definition}

\definecolor{MyBlue}{cmyk}{1,0.13,0,0.63}
\definecolor{MyGreen}{cmyk}{0.91,0,0.88,0.52}
\newcommand{\mylinkcolor}{MyBlue}
\newcommand{\mycitecolor}{MyGreen}
\newcommand{\myurlcolor}{black}

\usepackage{hyperref}
\hypersetup{%
  bookmarksnumbered=true,bookmarksopen=false,%
  plainpages=false,% necessary to prevent duplicate page identifiers
  linktocpage=true,%
  colorlinks=true,breaklinks=true,%
  linkcolor=\mylinkcolor,citecolor=\mycitecolor,urlcolor=\myurlcolor,%
  pdfpagelayout=OneColumn,%
  pageanchor=true,%
}

\begin{document}

\maketitle

\vspace{-2pc}

\begin{abstract}
We use  spectral flow to present a new proof of Levinson's theorem for Schr\"{o}dinger operators on $\R^n$ with smooth compactly supported potential. Our proof is valid in all dimensions and in the presence of resonances. The statement is expressed in terms of the spectral shift function and the ``high energy corrected time delay'' following Guillop\'{e} and others. 
\end{abstract}
\maketitle

\parindent=0.0in
\parskip=0.00in

%\tableofcontents

\parskip=0.06in

\section{Introduction}
\label{sec:intro}
Much work has been done in recent years investigating the topological nature of Levinson's theorem from quantum scattering theory, both as an index theorem \cite{ANRR, AR23-4D, kellendonk08, kellendonk12, richard13, richard13ii, richard21} and as an index pairing \cite{alexander24, alexander23thesis, AR23}.
In this paper we prove Levinson's theorem for Hamiltonians $H_0,H$ on $\R^n$ by using spectral flow from $H_0$ to $H$. By applying the operator pseudodifferential calculus to the spectral flow formula of \cite{CPotSuk}, we obtain a proof of the integral form of Levinson's theorem in all dimensions and in the presence of zero energy resonances. The dominant contribution is from the eta invariants of the endpoints $H_0,H$, and  can be computed using the Birman-Kre\u{\i}n formula.
%The topological approach we provide in this paper is of a  different nature, expressing Levinson's theorem as the result of the spectral flow along a path of Hamiltonians. 
In particular, we give a new approach to the relationship between the spectral shift function and spectral flow, extending work of Azamov, Carey, Dodds and Sukochev \cite{ACDS, ACS}.
Our main result (see Theorem \ref{thm:levinson}) is

\begin{thm*}[Levinson's theorem]
Suppose that $V \in C_c^\infty(\R^n)$. Then the number $N$ of eigenvalues (counted with multiplicity) of $H = H_0+V$ is given by
\begin{align*}
-N &= \frac{1}{2\pi i}\int_0^\infty \left(\textup{Tr}(S(\lambda)^*S'(\lambda)) - p_n(\lambda) \right) \, \d \lambda - \beta_n(V) +N_{res}
\end{align*}
where $N_{res}$ is the contribution from resonances as defined in Theorem \ref{thm:spec-shift-values}
and the polynomial $p_n$ and constant $\beta_n$ are defined in Lemma \ref{lem:consts-beta} and Definition \ref{defn:high-energy-poly1}.
\end{thm*}

The layout of the paper is as follows. In Section \ref{sec:recap-spec-flow} we recall the definition of spectral flow due to Phillips \cite{phillips96, phillips97} and the general formula for the spectral flow along a path of unbounded operators from \cite{CPotSuk}.  In Section \ref{sec:stat-scat} we summarise the stationary scattering theory for the Hamiltonians $H_0,H$ and in Section \ref{sec:spec-shift} we recall the spectral shift functon and its defining properties, including the Birman-Kre\u{\i}n trace formula. In Section \ref{sec:high-energy} we describe the high-energy behaviour of the spectral shift function from \cite{alexander24} and the pseudodifferential expansion of the resolvent from \cite{CPRS}.

In Section \ref{sec:spec-flow} we use the scattering techniques of Section \ref{sec:background} to analyse the spectral flow formula in two components. The first is an `integral of one-form' type term in Section \ref{sec:one-form} and the second is a Birman-Kre\u{\i}n contribution in Section \ref{sec:eta-term}. Finally, in Section \ref{sec:levinson} we obtain a formula for the spectral flow in terms of scattering data and as a consequence prove Levinson's theorem.

{\bf Acknowledgements} This work is supported by the ARC Discovery grant DP220101196 and AA was supported by an Australian Government RTP scholarship.
We would like to thank Alan Carey and Galina Levitina for useful input.

%%%%%%%%%%%%%%%%%%%%%%%%%%%%%%%%%%%%%%%%%%%%%%%%%%%%%%%%%%%%%%%%%%%%%%%%%%%
%%%%%%%%%%%%%%%%%%%%%%%%%%%%%%%%%%%%%%%%%%%%%%%%%%%%%%%%%%%%%%%%%%%%%%%%%%%
%%%%%%%%%%%%%%%%%%%%%%%%%%%%%%%%%%%%%%%%%%%%%%%%%%%%%%%%%%%%%%%%%%%%%%%%%%%
%%%%%%%%%%%%%%%%%%%%%%%%%%%%%%%%%%%%%%%%%%%%%%%%%%%%%%%%%%%%%%%%%%%%%%%%%%%

\section{Background and notations}\label{sec:background}

\subsection{Spectral flow}\label{sec:recap-spec-flow}

The concept of spectral flow was used by Atiyah, Patodi and Singer in \cite{aps75, aps76} as a tool to develop APS index theory. Spectral flow is intuitively defined as the net number of eigenvalues which change sign along a path of self-adjoint operators, with the convention that an eigenvalue changing from negative to positive will provide a contribution of $1$ to the spectral flow. We use the definition due to Phillips \cite{phillips96, phillips97}. Phillips' definition of spectral flow is valid in the much broader setting of semifinite von Neumann algebras with faithful normal semifinite traces, and while we do not need the full power of such a definition, we do need the ability to handle operators with continuous spectrum.

Consider the compact operators $\mathcal{K}(\H) \subset \B(\H)$ with trace $\Tr$ and let $\pi: \B(\H) \to \B(\H)\slash \mathcal{K}(\H)$ denote the projection onto the Calkin algebra. Let $\chi = \chi_{[0,\infty)}$ be the characteristic function of the interval $[0,\infty)$. Let $(T_t)_{t \in [0,1]}$ be any norm-continuous path of bounded self-adjoint Fredholm operators in $\B(\H)$, so that $\pi(T_t)$ is a norm continuous path of invertibles. Then $\pi\left(\chi\left(T_t \right) \right) = \chi \left( \pi \left(T_t \right) \right)$. Since the spectrum of the $\pi\left(T_t\right)$ are bounded away from zero, the path $\chi\left(\pi\left(T_t \right) \right)$ is continuous. By compactness (and \cite[Lemma 4.1]{BCPRSW}) we can choose a partition $0 = t_0 < t_1 < \cdots < t_k = 1$ such that
\begin{align*}
\left|\left| \pi \left( \chi \left( T_t \right) \right) - \pi \left( \chi \left(T_s \right) \right) \right| \right| &< \frac12
\end{align*}
for all $t,s \in [t_{i-1}, t_i]$ and $1 \leq i \leq k$. Defining the projection $P_i = \chi\left(T_{t_i} \right)$ we find that $P_{i-1}P_i: P_i \H \to P_{i-1} \H$ is Fredholm. We recall the following definition, due to Phillips \cite{phillips96, phillips97}.
\begin{defn}
Let $\H$ be a Hilbert space. For $t \in [0,1]$ let $(T_t)$ be any norm-continuous path of bounded self-adjoint Fredholm operators in $\B(\H)$. For a partition $0 = t_0 < t_1 < \cdots < t_k = 1$ of the interval $[0,1]$  define the operators $P_i = \chi\left(T_{t_i} \right)$. Then we define the spectral flow of the path $(T_t)$ by
\begin{align*}
\textup{sf}(T_t) &:= \sum_{i=1}^k \textup{Index}\left( P_{i-1} P_i \right).
\end{align*}
\end{defn}
We note that the above definition of spectral flow is independent of the choice of partition \cite{lesch05, phillips96, phillips97} and agrees with the topological definition used in \cite{aps75, aps76} when both make sense. For unbounded operators, we make the following definition of spectral flow \cite{CP1}.

\begin{defn}
Let $\H$ be a Hilbert space with trace $\Tr$. Let $(D_t)$ be a graph norm continuous path of unbounded self-adjoint Fredholm operators on $\H$. Define the function $F:\R \to [-1,1]$ by $F(x) = x(1+x^2)^{-\frac12}$. The spectral flow along the path $(D_t)$ is defined by
\begin{align*}
\textup{sf}(D_t) &:= \textup{sf}(F(D_t)).
\end{align*}
\end{defn}
Throughout the rest of this section $[0,1] \ni t \mapsto D_t$ stands for a path of unbounded self-adjoint linear Fredholm operators acting on some dense domain in $\H = L^2(\R^n)$. We denote by $(F_t) = (F(D_t))$ the bounded transform of the path $(D_t)$. We must also impose a smoothness assumption on $D_t$ to use analytic formulae for the spectral flow.

\begin{defn}
\begin{enumerate}
	\item A path $[0,1] \ni t \mapsto D_t$ is called $\Gamma$-differentiable at the point $t = t_0$ if and only if there exists a bounded linear operator $T$ such that
\begin{align*}
\lim_{t \to t_0} \left| \left| t^{-1} (D_t-D_{t_0})(\textup{Id}+D_{t_0}^2)^{-\frac12} - T \right| \right| &= 0.
\end{align*}
In this case we set $\dot{D}_{t_0} = T(\textup{Id}+D_{t_0}^2)^{\frac12}$. The operator $\dot{D}_t$ is a symmetric linear operator with domain $\Dom(D_t)$ \cite[Lemma 25]{CPotSuk}.
	\item If the mapping $t \mapsto \dot{D}_t (\textup{Id}+D_t^2)^{-\frac12}$ is defined and continuous with respect to the operator norm, then we call the path $t \mapsto D_t$ a continuously $\Gamma$-differentiable or a $C_\Gamma^1$ path.
\end{enumerate}
\end{defn}

The most general analytic spectral flow formula for the case of unbounded operators on a Hilbert space is given by the following theorem \cite[Theorem 9]{CPotSuk}. The sign of the second term in \eqref{eq:ceepotsuk} below appears incorrectly in \cite[Theorem 9]{CPotSuk}.

\begin{thm}\label{thm:cpotsuk}
Let $[0,1] \ni t \mapsto D_t$ be a piecewise $C_\Gamma^1$ path of linear operators and $F_t \in \B(\H)$ be Fredholm with $\norm{F_t} \leq 1$. Let $g:\R \to \R$ be a positive $C^2$ function such that
\begin{enumerate}
	\item $\int_\R g(x) \, \d x = 1$;
	\item $\int_0^1 \left|\left| \dot{D}_t g(D_t) \right| \right|_1 \, \d t < \infty$; and
	\item $G(D_1)-\frac12 B_1 - G(D_0) + \frac12 B_0 \in \mathcal{L}^1(\H)$, where $B_j = 2 \chi_{[0,\infty)}(D_j)-1$, and $G$ is the antiderivative of $g$ such that $G(\pm \infty) = \pm \frac12$. 
\end{enumerate}
Then
\begin{align}
\textup{sf}(D_t) &= \int_0^1 \Tr \left( \dot{D}_t g(D_t) \right) \, \d t - \Tr \left( G(D_1)-\frac12 B_1 - G(D_0) + \frac12 B_0 \right).
\label{eq:ceepotsuk}
\end{align}
\end{thm}
In our applications of this formula we will take $D_t=H_0+\alpha{\rm Id} +tV$ where $H_0$ is the free Hamiltonian, $\alpha$ a carefully chosen constant and $V$ a suitable potential. We now describe these ingredients.

\subsection{Stationary scattering theory}\label{sec:stat-scat}

We consider the scattering theory on $\R^n$ associated to the operators
\begin{align*}
H_0  &= - \sum_{j=1}^n \frac{\partial^2}{\partial x_j^2} = - \Delta \quad \textup{ and } H = H_0+V,
\end{align*}
where the (multiplication operator by the) potential $V$ is a smooth compactly supported and real-valued function. With $\langle\cdot,\cdot\rangle$ the Euclidean inner product on $\R^n$, we denote the Fourier transform by
\[
\F:L^2(\R^n)\to L^2(\R^n),\qquad [\F f](\xi)=(2\pi)^{-\frac{n}{2}}\int_{\R^n}e^{-i\langle x,\xi\rangle}f(x)\,\d x.
\]
Note that the Fourier transform $\F$ is an isomorphism from $H^{s,t}$ to $H^{t,s}$ for any $s,t \in \R$.

We denote by $\mathcal{B}(\H_1,\H_2)$, $\mathcal{K}(\H_1,\H_2)$ and $\L^1(\H_1, \H_2)$ the bounded, compact and trace class operators from $\H_1$ to $\H_2$. For $z \in \C \setminus \R$, we let
\[
R_0(z)=(H_0-z)^{-1},\qquad R(z)=(H-z)^{-1}.
\]
The operator $H_0$ has purely absolutely continuous spectrum, and in particular no kernel. The operator $H$ can have finitely many eigenvalues which are negative, or zero \cite[Theorem 6.1.1]{yafaev10}. 

Several Hilbert spaces recur, and we adopt the notation (following \cite[Section 2]{jensen81} which contains an excellent discussion on the relations between the spaces and operators we introduce here)
\[
\H = L^2(\R^n),\quad \mathcal{P} = L^2(\Sf^{n-1}),\quad \H_{spec} = L^2(\R^+, \mathcal{P}) \cong L^2(\R^+) \otimes \mathcal{P}.
\] 
Here $\H_{spec}$ provides the Hilbert space on which we can diagonalise the free Hamiltonian $H_0$.

Since $V$ is bounded, $H = H_0+V$ is self-adjoint with $\Dom(H) = \Dom(H_0)$. Since $V\in C_c^\infty(\R^n)$, the wave operators 
\[
W_\pm=\mathop{\textup{s-lim}}_{t\to\pm\infty}e^{itH}e^{-itH_0}
\]
exist and are asymptotically complete \cite[Theorem 1.6.2]{yafaev10}. 
The wave operators are partial isometries satisfying $W_\pm^*W_\pm = \textup{Id}$ and $W_\pm W_\pm^* = P_{ac}$, the projection onto the absolutely continuous subspace for $H$. The scattering operator is the unitary operator
\begin{align}
S &= W_+^*W_-,
\end{align}
which commutes strongly with the free Hamiltonian $H_0$.
For our analysis of the scattering operator, we describe the explicit unitary which diagonalises the free Hamiltonian. 
\begin{defn}
\label{def:diag}
Define the operator which diagonalises the free Hamiltonian $H_0$ as
\[
F_0: \H \to \H_{spec}\quad \mbox{by} \quad [F_0 f](\lambda,\omega) = 2^{-\frac12} \lambda^{\frac{n-2}{4}} [\F f](\lambda^\frac12 \omega).
\] 
\end{defn}
By \cite[p. 439]{jensen81} the operator $F_0$ is unitary and for $\lambda \in[0,\infty)$, $\omega\in \Sf^{n-1}$ and $f \in \H_{spec}$ we have
\[
[F_0H_0F_0^* f](\lambda,\omega)=\lambda f(\lambda,\omega)=:\lambda f(\lambda,\omega).
\] 
As a consequence of the relation $SH_0 = H_0 S$, there exists a family $\{S(\lambda) \}_{\lambda \in \R^+}$ of unitary operators on $\mathcal{P} = L^2(\Sf^{n-1})$ such that for all $\lambda \in \R^+, \omega \in \Sf^{n-1}$ and $f \in \H$ we have
\begin{align*}
[F_0 S f](\lambda,\omega) &= S(\lambda) [F_0 f](\lambda,\omega).
\end{align*}
For historical reasons, we refer to $S(\lambda)$ as the scattering matrix at energy $\lambda \in \R^+$ since in dimension $n = 1$ the operator $S(\lambda)$ is an $M_2(\C)$-valued function.

Note that the operators $H_0$ and $H$ are not Fredholm, since $0$ is in the essential spectrum of both. To use the spectral flow formula of Theorem \ref{thm:cpotsuk} we make the following adjustment for the rest of this article. Let $\nu \leq 0$ be the furthest eigenvalue of $H$ from zero. We fix $\alpha > -2\nu+1$, so that the operators $H_0(\alpha) = H_0+\alpha$ and $H(\alpha) = H+\alpha$ define Fredholm operators. As a consequence, the path
\begin{align*}
[0,1] \ni t \mapsto H_0+tV+\alpha =: H_t(\alpha)
\end{align*}
defines a $C_\Gamma^1$ path of Fredholm operators with $\dot{H}_t(\alpha) = V$. The operator $H_0(\alpha)$ has purely absolutely continuous spectrum $\sigma(H_0(\alpha)) = [\alpha,\infty)$ and the operator $H(\alpha)$ has absolutely continuous spectrum $\sigma_{ac}(H(\alpha)) = \sigma(H_0(\alpha))$. In addition, the operator $H(\alpha)$ has a finite number of distinct eigenvalues $0 < \lambda_1(\alpha) < \lambda_2(\alpha) < \cdots < \lambda_K(\alpha) \leq \alpha$ of finite multiplicity. The eigenvalues satisfy $\lambda_j(\alpha) = \lambda_j + \alpha$, with $\lambda_1 < \lambda_2 < \cdots < \lambda_K \leq 0$ the distinct eigenvalues of $H$. We write $M(\lambda_j) = M(\lambda_j(\alpha))$ for the multiplicity of the eigenvalue $\lambda_j$ and use the notation $N_0$ for the multiplicity of the zero eigenvalue for $H$. We also write
\begin{align*}
N &= \sum_{j=1}^K M(\lambda_j)
\end{align*}
for the total number of eigenvalues of $H$ (counted with multiplicity). Let $P_{ac}(H_0(\alpha))$ denote the projection onto the absolutely continuous spectrum for $H_0(\alpha)$. The wave operators 
\begin{align*}
W_\pm(\alpha) = \slim_{t \to \pm \infty} \e^{itH(\alpha)} \e^{-itH_0(\alpha)} P_{ac}(H_0(\alpha)) = W_\pm
\end{align*}
exist and are asymptotically complete by the invariance principle \cite[Theorem XI.11]{reedsimon79}. Direct calculation gives the following diagonalisation for $H_0(\alpha)$.
\begin{lemma}
The operator $F_\alpha: \H \to L^2([\alpha,\infty)) \otimes L^2(\Sf^{n-1})$ given by 
\begin{align*}
[F_\alpha f](\lambda,\omega) = [F_0 f](\lambda-\alpha,\omega)
\end{align*}
satisfies
\begin{align*}
[F_\alpha H_0(\alpha) f](\lambda,\omega) = \lambda [F_\alpha f](\lambda,\omega).
\end{align*}
\end{lemma}
The scattering operator $S = W_+^*W_-$ is unitary and commutes with $H_0(\alpha)$ and so there exists a family $\{S_\alpha(\lambda)\}_{\lambda \in [\alpha,\infty)}$  of unitary operators on $L^2(\Sf^{n-1})$ such that
\begin{align*}
[F_\alpha S f](\lambda,\omega) &= S_\alpha(\lambda) [F_\alpha f](\lambda,\omega).
\end{align*}
In fact, we have $S_\alpha(\lambda) = S(\lambda-\alpha)$ for all $\lambda \in [\alpha,\infty)$. Pointwise we have $S_\alpha(\lambda)-\textup{Id} \in \L^1(L^2(\mathbb{S}^{n-1}))$, \cite[Proposition 8.1.5]{yafaev10}.

\subsection{The spectral shift function and the Birman-Kre\u{\i}n trace formula}\label{sec:spec-shift}

We now recall the spectral shift function \cite{birman62, krein53} for the pair $(H(\alpha),H_0(\alpha))$ and some of its defining properties (see \cite[Theorems 0.9.2 and 0.9.7]{yafaev10}). The proofs in \cite{yafaev10} only consider $\alpha = 0$, however extend directly to $\alpha > 0$ by translation.

\begin{thm}\label{thm:ssf-properties}
Suppose that $V \in C_c^\infty(\R^n)$, $\alpha \geq 0$ and let $S$ be the corresponding scattering operator. Then there exists a unique (up to an additive constant) real-valued piecewise-$C^1$ function $\xi_\alpha(\cdot,H,H_0): \R \to \R$ such that 
\begin{align}\label{eq:ssf-defining}
\textup{Tr}(f(H(\alpha))-f(H_0(\alpha))) &=  \int_\R \xi_\alpha(\lambda, H, H_0) f'(\lambda) \, \d \lambda,
\end{align}
at least for all $f \in C^2(\R)$ with two locally bounded derivatives and satisfying
\begin{align}\label{eq:ssf-condition}
\frac{\d}{\d \lambda} \left( \lambda^{m+1} f'(\lambda) \right) &= O(\lambda^{-1-\eps})
\end{align}
as $\lambda \to \infty$, for some $\eps > 0$ and $m > \frac{n}{2}$.
We specify $\xi_\alpha(\cdot, H, H_0)$ uniquely by the convention $\xi_\alpha(\lambda, H, H_0) = 0$ for $\lambda$ sufficiently negative. Thus for $\lambda < \alpha$, $\xi_\alpha(\cdot,H,H_0)$ satisfies  the relation
\begin{align*}
\xi_\alpha(\lambda,H,H_0) &= -\sum_{k=1}^K M(\lambda_k(\alpha))\, \chi_{[\lambda_k(\alpha),\infty)}(\lambda),
\end{align*}
where we we have indexed the distinct eigenvalues of $H(\alpha)$ as $\lambda_1(\alpha) < \cdots < \lambda_K(\alpha)$ and each $\lambda_j(\alpha)$ has multiplicity $M(\lambda_j(\alpha))$. Furthermore, we have $\xi_\alpha(\cdot,H,H_0) \vert_{(\alpha,\infty)} \in C^1(\alpha,\infty)$ and for $\lambda > \alpha$ the relations
\begin{align*}
\textup{Det}(S_\alpha(\lambda)) &= \e^{-2\pi i \xi_\alpha(\lambda)} \quad \textup{ and } \quad \textup{Tr}\left(S_\alpha(\lambda)^*S_\alpha'(\lambda)\right) = -2\pi i \xi_\alpha'(\lambda)
\end{align*}
hold. Furthermore, we have $\xi_\alpha(\lambda) = \xi_0(\lambda-\alpha)$ for almost all $\lambda \in \R$.
\end{thm}

We call $\xi_\alpha(\cdot,H,H_0)$ the spectral shift function for the pair $(H(\alpha),H_0(\alpha))$ and will often just write $\xi_\alpha = \xi_\alpha(\cdot,H,H_0)$. We also write $\xi = \xi_0$. Using integration by parts we can rewrite the definining property \eqref{eq:ssf-defining} in a sometimes more convenient fashion, known as the Birman-Kre\u{\i}n trace formula \cite[Theorem III.4]{guillope81}.

\begin{lemma}
Suppose that $V \in C_c^\infty(\R^n)$, $\alpha \geq 0$ and let $S_\alpha, \xi_\alpha$ be the corresponding scattering operator and spectral shift function. Then for all $f \in C_c^\infty(\R)$ we have
\begin{align*}
\textup{Tr}(f(H(\alpha))-f(H_0(\alpha))) &= \frac{1}{2\pi i} \int_\alpha^\infty f(\lambda) \textup{Tr}\left(S_\alpha(\lambda)^*S_\alpha'(\lambda) \right) \, \d \lambda  + \sum_{k=1}^K f(\lambda_k(\alpha)) M(\lambda_k(\alpha)) \\
&\quad + f(\alpha) \left( \xi_\alpha(\alpha-)-\xi_\alpha(\alpha+)-M(\alpha) \right),
\end{align*}
where we have defined $\displaystyle \xi_\alpha(\alpha \pm) = \lim_{\eps \to 0^+} \xi_\alpha(\alpha \pm \varepsilon)$.
\end{lemma}

In fact by Theorem \ref{thm:ssf-properties} we have, with $N$ the total number of eigenvalues of $H$ counted with multiplicity and $N_0 = M(\alpha)$ the number of zero eigenvalues for $H$, the relation $\xi_\alpha(\alpha-) = -N+N_0$. We can then rewrite the Birman-Kre\u{\i}n trace formula as
\begin{align*}
\textup{Tr}(f(H(\alpha))-f(H_0(\alpha))) &= \frac{1}{2\pi i} \int_\alpha^\infty f(\lambda) \textup{Tr}\left(S_\alpha(\lambda)^*S_\alpha'(\lambda) \right) \, \d \lambda  + \sum_{k=1}^K f(\lambda_k(\alpha)) M(\lambda_k(\alpha)) \\
&\quad + f(\alpha) \left( -N-\xi_\alpha(\alpha+) \right).
\end{align*}

\subsection{Resolvent expansions and limiting behaviour of the spectral shift function}\label{sec:high-energy}

For $k \in \N \cup \{0\}$ and $f \in C_c^\infty(\R^n)$ we introduce the notation $f^{(k)} = [H_0, [H_0, [\cdots, [H_0, f] \cdots ]] ]$, where the expression has $k$ commutators. We recall the following pseudodifferential expansion of the resolvent \cite[Lemma 4.8]{alexander24} (see also \cite[Lemma 6.11]{CPRS}).
\begin{lemma}\label{lem:psido}
Suppose that $V \in C_c^\infty(\R^n)$. For all $M, K \geq 0$ and $z \notin \sigma(H)$ we have the expansion
\begin{align*}
R(z)=(H-z)^{-1} &= \sum_{m=0}^M \left( \sum_{|k|=0}^K C_m(k) (-1)^{m+|k|} V^{(k_1)} \cdots V^{(k_m)} R_0(z)^{m+|k|+1} +P_{m,K}(z) \right) \\
& \quad + (-1)^{M+1} (R_0(z)V)^{M+1} R(z),
\end{align*}
where the remainder $P_{m,K}(z)$ is a pseudodifferential operator of order at most $-2m-K-3$. The combinatorial coefficients $C_m(k)$ are given by
\begin{align*}
C_m(k) &= \frac{(m+|k|)!}{k_1! \cdots k_m ! (k_1+1) (k_1+k_2+2) \cdots (|k|+m)}.
\end{align*}
\end{lemma}

Note that the operator $V^{(k_1)} \cdots V^{(k_m)}$ is a differential operator of order at most $|k|$ with smooth compactly supported coefficients and thus we may write
\begin{align}\label{eq:psido-notation}
V^{(k_1)} \cdots V^{(k_m)} &= \sum_{|\beta| = 0}^{|k|} g_{k,\beta} \partial^\beta,
\end{align} 
where the multi-indices $\beta$ are of length $n$ and $g_{k,\beta} \in C_c^\infty(\R^n)$.

We now recall the high-energy behaviour of the spectral shift function and its derivative \cite[Lemma 2.15, Theorem 4.15 and Remark 4.16]{alexander24}.

\begin{lemma}
\label{lem:consts-beta}
Suppose that $V \in C_c^\infty(\R^n)$. Then for $1 \leq \ell \leq \lfloor \frac{n}{2} \rfloor$ there exist coefficients $C_\ell(n,V), c_\ell(n,V), \beta_n(V)$ such that
\begin{align*}
0 &= \lim_{\lambda \to \infty} \left(-2\pi i \xi(\lambda) - 2\pi i \beta_n(V)-\sum_{\ell=1}^{\lfloor \frac{n-1}{2} \rfloor} C_\ell(n,V) \lambda^{\frac{n}{2}-\ell} \right) \\
&= \lim_{\lambda \to \infty} \left( -2\pi i\xi'(\lambda) - \sum_{\ell=1}^{\lfloor \frac{n-1}{2} \rfloor} c_\ell(n,V) \lambda^{\frac{n}{2}-\ell-1} \right).
\end{align*}
The coefficients are related by $c_\ell(n,V) = \left(\frac{n}{2}-\ell \right) C_\ell(n,V)$. For $1 \leq \ell \leq \lfloor \frac{n-1}{2} \rfloor$ and $M, K \in \N$ with $M+K \geq n$ we define the set
\begin{align*}
Q_{M,K}(\ell) &= \big\{(m,k,\beta) \in \{ 0,1, \dots, M\} \times \{0,1,\dots,K\}^m \times \{0,1, \dots, K\}^n : |\beta| \leq |k|, \\
&\quad \textup{ and } m+|k|+1- \frac{|\beta|}{2} = \ell \big\}.
\end{align*}
The coefficients $C_\ell(n,V)$ are given by
\begin{align}\label{eq:high-energy-coefficients}
C_\ell(n,V)&= \sum_{(m, k,\beta)\in Q_{M,K}(j)} \!\frac{(-1)^{m+|k|+1}(2\pi i)C_{m}(k) (-i)^{|\beta|} \Gamma\left( \frac{\beta_1+1}{2} \right) \cdots \Gamma \left(\frac{\beta_n+1}{2} \right)}{(m+1) (m+|k|)! \Gamma\left( \frac{n}{2} - m - |k|+\frac{|\beta|}{2} \right)(2\pi)^n}\! \int_{\R^n} \!V(x) g_{k,\beta}(x) \, \d x,
\end{align}
and
\begin{align}\label{eq:beta-defn}
\beta_n(V) &= \begin{cases} 
0, \quad & \textup{ if } n \textup{ is odd}, \\
\frac{1}{2\pi i} C_{\frac{n}{2}}(n,V), \quad & \textup{ if } n \textup{ is even}.
\end{cases}
\end{align}
\end{lemma}

\begin{defn}\label{defn:high-energy-poly1}
Define the functions $P_n,p_n :(0,\infty) \to \C$ by
\begin{align*}
P_n(\lambda) &= 2\pi i\beta_n(V)+\sum_{\ell=1}^{\lfloor \frac{n-1}{2} \rfloor} C_\ell(n,V) \lambda^{\frac{n}{2}-\ell}, \\
p_n(\lambda) &=  \sum_{\ell=1}^{\lfloor \frac{n-1}{2} \rfloor} c_\ell(n,V) \lambda^{\frac{n}{2}-\ell-1} = P_n'(\lambda).
\end{align*}
We call $P_n$ the high-energy polynomial for $\xi$ and $p_n$ the high-energy polynomial for $\xi'$.
\end{defn}
\begin{rmk}
Recall the spectral shift functions $\xi, \xi_\alpha$ for the pairs $(H,H_0)$ and $(H(\alpha), H_0(\alpha))$. Since $\xi_\alpha(\lambda) = \xi(\lambda-\alpha)$ for all almost all $\lambda \in \R$ we have that the high-energy polynomial for $\xi_\alpha$ is $P_n(\cdot-\alpha)$ and likewise for $\xi'$ and $p_n(\cdot-\alpha)$.
\end{rmk}

We can explicitly compute the lowest order polynomials (see \cite{banuelos96, colin81}), finding $P_1 = 0$, and
\begin{align*}
P_2(\lambda) &= -\frac{ (2\pi i) \textup{Vol}(\Sf^1)}{2(2\pi)^2} \int_{\R^2} V(x) \, \d x = -\frac{2\pi i}{4\pi} \int_{\R^2} V(x) \, \d x, \\
P_3(\lambda) &= -\frac{(2\pi i) \lambda^\frac12 \textup{Vol}(\Sf^2)}{2(2\pi)^3} \int_{\R^3} V(x) \, \d x = -\frac{(2\pi i)\lambda^\frac12}{4 \pi^2} \int_{\R^3} V(x)\, \d x, \\
P_4(\lambda) &= - \frac{(2\pi i) \lambda \textup{Vol}(\Sf^3)}{2(2\pi)^4} \int_{\R^4} V(x)\,  \d x + \frac{(2\pi i) \textup{Vol}(\Sf^3)}{4(2\pi)^4} \int_{\R^4} V(x)^2\, \d x .
\end{align*}

The integrability properties of the derivative of the spectral shift function on $\R^+$ are well-known, see \cite[Theorem 5.2]{jia12} and \cite[Lemma 4.12]{alexander24}.

\begin{lemma}\label{lem:ssf-L1}
Suppose that $V \in C_c^\infty(\R^n)$. Then the function $\textup{Tr}\left(S(\cdot)^*S'(\cdot) \right) - p_n$ is integrable on $\R^+$. In particular, if $n = 1,2$ we have $\textup{Tr}\left(S(\cdot)^*S'(\cdot) \right) \in L^1(\R^+)$.
\end{lemma}

We now define zero-energy resonances, a low-energy phenomena known to provide obstructions to generic behaviour in scattering theory in low dimensions.

\begin{defn}\label{defn:resonances}
Suppose that $V \in C_c^\infty(\R^n)$. If $n \neq 2$ we say there is a resonance if there exists a non-zero bounded distributional solution to $H\psi = 0$. If $n = 2$ we say there is a $p$-resonance if there exists a non-zero distributional solution $\psi$ to $H\psi = 0$ with $\psi \in L^q(\R^2) \cap L^\infty(\R^2)$ for some $q > 2$. We say that there is an $s$-resonance if there exists a non-zero bounded distributional solution $\psi$ to $H\psi = 0$ with $\psi \notin L^q(\R^2)$ for all $q < \infty$.
\end{defn}
General bounds on the resolvent of $H$ \cite{jensen80} show that there can be no resonances for dimension $n \geq 5$.

We now recall the value of the spectral shift function at zero in all dimensions from \cite[Corollary 5.11]{alexander24}.
\begin{thm}\label{thm:spec-shift-values}
Suppose $V \in C_c^\infty(\R^n)$. Then the value of the spectral shift function at zero is given by $\xi(0+) = -N-N_{res}$, where $N_{res} = 0$ unless
\begin{align*}
N_{res} &= \begin{cases}
\frac12, \quad & \textup{ if } n = 1 \textup{ and there are no resonances}, \\
N_p, \quad & \textup{ if } n = 2 \textup{ and there are } N_p =0,1, 2 \textup{ } p\textup{-resonances}, \\
\frac12, \quad & \textup{ if } n = 3 \textup{ and there are resonances}, \\
1, \quad & \textup{ if } n = 4 \textup{ and there are resonances}.
\end{cases}
\end{align*}
\end{thm}
We note that the proof of Theorem \ref{thm:spec-shift-values} in \cite{alexander24} is as a corollary of Levinson's theorem, however the result can be obtained directly using perturbation determinant methods in odd dimensions (see \cite{levinson49}, \cite{newton77} and \cite[Theorem 3.3]{guillope81}).

%%%%%%%%%%%%%%%%%%%%%%%%%%%%%%%%%%%%%%%%%%%%%%%%%%%%%%%%%%%%%%%%%%%%%%%%%%%
%%%%%%%%%%%%%%%%%%%%%%%%%%%%%%%%%%%%%%%%%%%%%%%%%%%%%%%%%%%%%%%%%%%%%%%%%%%
%%%%%%%%%%%%%%%%%%%%%%%%%%%%%%%%%%%%%%%%%%%%%%%%%%%%%%%%%%%%%%%%%%%%%%%%%%%
%%%%%%%%%%%%%%%%%%%%%%%%%%%%%%%%%%%%%%%%%%%%%%%%%%%%%%%%%%%%%%%%%%%%%%%%%%%

\section{Spectral flow for Schr\"{o}dinger operators}\label{sec:spec-flow}
In this section we analyse the spectral flow formula of Theorem \ref{thm:cpotsuk} applied to the path $H_t(\alpha)$ by making a particular choice of the function $g$ and then taking residues. 

Define for $\textup{Re}(s) > \frac12$ the constants $C_s =\int_\R (1+u^2)^{-s} \, \d u$. The functions $s \mapsto C_s$ have a pole at $s =\frac12$ with residue equal to one.  For $\textup{Re}(s) > \frac12$ we define the eta function $\eta_s: \R \to \C$ by
\begin{align*}
\eta_s(x) &= \frac{1}{C_s} \int_1^\infty x(1+wx^2)^{-s} w^{-\frac12} \, \d w = \frac{2}{C_s} \int_x^\infty (1+v^2)^{-s} \, \d s,
\end{align*}
where the second expression is valid only for $x > 0$.
We can now use the function $\eta_s$ to obtain a useful form of Theorem \ref{thm:cpotsuk}.
\begin{lemma}
\label{lem:spec-flow-residue}
Let $[0,1] \ni t \mapsto D_t$ be a piecewise $C_\Gamma^1$ path of linear operators with $\dot{D}_t(1+D_t^2)^{-s}$ trace-class for all $s>\frac{n}{4}$.
%and $F_t \in \B(\H)$ be Fredholm with $\norm{F_t} \leq 1$.
Then
\begin{align}\label{eq:spec-flow-residue}
&\textup{sf}(D_t) = \nonumber \\
& \Res{s=\frac12} \left( \int_0^1 \Tr \left( \dot{D}_t (\textup{Id}+D_t^2)^{-s} \right) \d t +\frac{C_s}{2} \Tr \left( \eta_s(D_1)-\eta_s(D_0)+P_{\textup{Ker}(D_1)} - P_{\textup{Ker}(D_0)} \right) \right).
\end{align}
\end{lemma}
\begin{proof}
For $s > \frac n4$, let $g_s : \R \to \R$ be given by $g_s(x) = C_s^{-1}(1+x^2)^{-s}$. Note that the antiderivative $G_s$ of $g_s$ with $G(\pm \infty) = \pm \frac12$ is given by
\begin{align*}
G_s(x) &= -\frac12 + \frac{1}{C_s} \int_{-\infty}^x (1+u^2)^{-s} \, \d u.
\end{align*}
The function $g_s$ is even and so $G_s$ is odd. For $x > 0$ we have $G_s(x) = \frac12-\frac12 \eta_s(x)$, while for $x<0$ we have $G_s(x)=-\frac12-\frac12 \eta_s(x)$. Thus applying Theorem \ref{thm:cpotsuk} to $g_s$, $G_s$ and multiplying both sides by $C_s$ yields
\begin{align}\label{eq:meromorphic}
C_s \textup{sf}(D_t) &= \int_0^1 \Tr \left( \dot{D}_t (\textup{Id}+D_t^2)^{-s} \right) \d t + \frac{C_s}{2} \Tr \left( \eta_s(D_1)-\eta_s(D_0)+P_{\textup{Ker}(D_1)} - P_{\textup{Ker}(D_0)} \right).
\end{align}
The left-hand side of Equation \eqref{eq:meromorphic} is a meromorphic function of $s$ with a simple pole at $s = \frac12$ and thus defines a meromorphic continuation of the right-hand side of Equation \eqref{eq:meromorphic} with a simple pole at $s = \frac12$. As a result, taking the residue at $s = \frac12$ gives Equation \eqref{eq:spec-flow-residue}.
\end{proof}
Equation \eqref{eq:spec-flow-residue} is the starting point for our analysis of the spectral flow along the path $H_t(\alpha)$. There are two separate types of terms to be considered. The first is  the ``integral of one-form'' term which is evaluated in Section \ref{sec:one-form} using the pseudodifferential expansion of Lemma \ref{lem:psido} and the second is the $\eta$ contribution which is evaluated in Section \ref{sec:eta-term} using the Birman-Kre\u{\i}n trace formula.

\subsection{The ``integral of one form'' term}\label{sec:one-form}

We use the pseudodifferential expansion of Lemma \ref{lem:psido} to compute an expansion for the integral of one form term in Theorem \ref{thm:cpotsuk}. After a fixed number of terms (depending on the dimension $n$) the remainder term will be holomorphic at $s = \frac12$ and can be discarded. We begin with a residue computation.

\begin{lemma}\label{lem:even-one-form-residue}
For $\ell \in \N$, $\alpha > 0$ we have
\begin{align*}
\mathop{\textup{Res}}_{s = \frac12} \left( \int_0^\infty u^{\ell-1} (1+(u+\alpha)^2)^{-s} \, \d u \right) &= \mathop{\sum_{j=0}^{\ell-1}}_{j \textup{ even}} \binom{\ell-1}{j} \frac{(-1)^{\ell-\frac{j}{2}-1} \alpha^{\ell-j-1} \Gamma\left(\frac{j+1}{2}\right)}{4\Gamma\left(\frac{j}{2}+1\right)\Gamma\left(\frac12 \right)}.
\end{align*}
\end{lemma}
\begin{proof}
Fix $s \in \C$ with $\textup{Re}(s) > \frac{\ell+1}{2}$. We make the substitution $v = u+\alpha$ and apply the binomial expansion to obtain
\begin{align*}
\int_0^\infty u^{\ell-1} (1+(u+\alpha)^2)^{-s} \, \d u  &= \int_\alpha^\infty (v-\alpha)^{\ell-1} (1+v^2)^{-s} \, \d v \\
&= \sum_{j=0}^{\ell-1} \binom{\ell-1}{j} (-\alpha)^{\ell-j-1} \int_\alpha^\infty v^j(1+v^2)^{-s} \, \d v \\
&=  \sum_{j=0}^{\ell-1} \binom{\ell-1}{j} (-\alpha)^{\ell-j-1} \int_0^\infty v^j(1+v^2)^{-s} \, \d v \\
&\quad -  \sum_{j=0}^{\ell-1} \binom{\ell-1}{j} (-\alpha)^{\ell-j-1} \int_0^\alpha v^j(1+v^2)^{-s} \, \d v.
\end{align*}
Since the integrals from $0$ to $\alpha$ are over a finite region, they are holomorphic at $s = \frac12$ and thus have vanishing residue. So we compute for $0 \leq j \leq \ell-1$ that
\begin{align*}
 \int_0^\infty v^j(1+v^2)^{-s} \, \d v &= \frac12 \int_0^\infty w^{\frac{j+1}{2}-1} (1+w)^{-s} \, \d w 
= \frac{\Gamma\left(\frac{j+1}{2} \right) \Gamma \left(s-\frac{j+1}{2} \right)}{2 \Gamma(s)}.
\end{align*}
Taking the residue at $s = \frac12$ we find 
\begin{align*}
\mathop{\textup{Res}}_{s=\frac12} \left( \int_0^\infty v^j(1+v^2)^{-s} \, \d v \right) &= \begin{cases}
\frac{(-1)^{\frac{j}{2}}{\frac{j+1}{2}}\Gamma\left(\frac{j+1}{2} \right)}{2\Gamma\left(\frac{j}{2}+1 \right) \Gamma\left(\frac12 \right)}, \quad &\textup{if } j \textup{ is even}, \\
0, \quad &\textup{otherwise},
\end{cases}
\end{align*}
from which the result follows.
\end{proof}

To evaluate some further traces, we need to be able to integrate polynomials over $\mathbb{S}^{n-1}$. We use the following result \cite{folland}.
\begin{lemma}\label{lem:folland}
Let $\beta$ be a multi-index of length $n$ and let $P_\beta:\R^n \to \C$ be given by $P_\beta(x) = x^\beta = x_1^{\beta_1} \cdots x_n^{\beta_n}$. Then
\begin{align*}
\int_{\Sf^{n-1}} P_\beta(\omega)\, \d \omega &= \begin{cases} 0, \quad & \textup{ if some } \beta_j \textup{ is odd,} \\
\frac{2 \Gamma\left( \frac{\beta_1+1}{2} \right) \cdots \Gamma\left( \frac{\beta_n+1}{2} \right)}{\Gamma\left( \frac{n+|\beta|}{2} \right)}, \quad & \textup{ if all } \beta_j \textup{ are even.}
\end{cases}
\end{align*}
\end{lemma}

We now contribute the residue of the contribution from the ``integral of one form'' term to the spectral flow.

\begin{prop}\label{prop:one-form-res}
Suppose that $V \in C_c^\infty(\R^n)$ and $\alpha > -2\nu$. Then for $n$ odd we have
\begin{align*}
\mathop{\textup{Res}}_{s=\frac12} \left( \int_0^1 \textup{Tr}\left(V\left( \textup{Id}+(H_0+tV+\alpha)^2 \right)^{-s} \right) \, \d t \right) &=0.
\end{align*}
If $n$ is even we have
\begin{align*}
&\mathop{\textup{Res}}_{s=\frac12} \left( \int_0^1 \textup{Tr}\left(V\left( \textup{Id}+(H_0+tV+\alpha)^2 \right)^{-s} \right) \, \d t \right) \\
&=\sum_{\ell = 1}^{\frac{n}{2}} \mathop{\sum_{j=0}^{\frac{n}{2}-\ell}}_{j \textup{ even}} \binom{\frac{n}{2}-\ell}{j} \frac{(-1)^{\frac{n}{2}-\ell-\frac{j}{2}-1} \alpha^{\frac{n}{2}-\ell-j}\Gamma\left(\frac{j+1}{2}\right)}{4(2\pi i)\Gamma\left(\frac{j}{2}+1\right)\Gamma\left(\frac12 \right)} C_\ell(n,V),
\end{align*}
with the $C_\ell(n,V)$ the high-energy coefficients defined in Equation \eqref{eq:high-energy-coefficients}.
\end{prop}
\begin{proof}
For $\alpha > 0$ and $\textup{Re}(s) > \frac{n}{4}$ we define the function $\varphi_{\alpha,s}: \R \to \C$ by
\begin{align*}
\varphi_{\alpha,s}(x) &= \left(1+(x+\alpha)^2 \right)^{-s},
\end{align*}
using the principal branch of the logarithm. The function $\varphi_{\alpha,s}$ is holomorphic in the half-plane $\textup{Re}(z) > - \alpha$.
Let $a \in \left(-\frac{\alpha}{2},0 \right)$ so that $a < \lambda$ for all $t \in [0,1]$ and $\lambda \in \sigma(H_0+tV)$ and define the vertical line $\gamma = \{ a+iv : v \in \R\}$.
For $t \in [0,1]$ we use Cauchy's integral formula to write
\begin{align}\label{eq:cauchy-int}
\varphi_{\alpha,s}(H_0+tV) &= - \frac{1}{2\pi i} \int_\gamma \varphi_{\alpha,s}(z) (H_0+tV -z)^{-1} \, \d z.
\end{align}
Denoting $R_t(z) = (H_0+tV-z)^{-1}$ we have by Lemma \ref{lem:psido} that for all $M, K \geq 0$ that 
\begin{align*}
R_t(z) &= \sum_{m=0}^M \left( t^m \sum_{|k| = 0}^K (-1)^{m+|k|}C_m(k) V^{(k_1)} \cdots V^{(k_m)} R_0(z)^{m+|k|+1} + P_{m,K,t}(z) \right) \\
&\quad + (-1)^{M+1} t^{M+1} (R_0(z)V)^{M+1} R_t(z),
\end{align*}
where $P_{m,K,t}(z)$ has order (at most) $-2m-K-3$. We can now write Equation \eqref{eq:cauchy-int} as 
\begin{align*}
\varphi_{\alpha,s}(H_0+tV) &= - \frac{1}{2\pi i} \sum_{m=0}^M \Bigg( t^m \sum_{|k|=0}^K (-1)^{m+|k|} C_m(k) V^{(k_1)} \cdots V^{(k_m)} \int_\gamma \varphi_{\alpha,s}(z) R_0(z)^{m+|k|+1}\, \d z \\
&\quad +\int_\gamma \varphi_{\alpha,s}(z) P_{m,K,t}(z) \, \d z \Bigg)  + \frac{(-1)^M t^{M+1}}{2\pi i} \int_\gamma \varphi_{\alpha,s}(z) (R_0(z) V)^{M+1} R_t(z) \, \d z \\
&:= - \frac{1}{2\pi i} \sum_{m=0}^M \ t^m \sum_{|k|=0}^K (-1)^{m+|k|} V^{(k_1)} \cdots V^{(k_m)} \int_\gamma \varphi_{\alpha,s}(z) R_0(z)^{m+|k|+1}\, \d z \\
&\quad + E(M,K,t,\alpha,s).
\end{align*}
Using again Cauchy's integral formula we can compute that 
\begin{align*}
\frac{1}{2\pi i} \int_\gamma \varphi_{\alpha,s}(z) R_0(z)^{m+|k|+1} \, \d z &= - \frac{1}{(m+|k|)!} \frac{\d^{m+|k|}\varphi_{\alpha,s}}{\d z^{m+|k|}} \vert_{z = H_0}.
\end{align*}
Thus we have the expression
\begin{align*}
\varphi_{\alpha,s}(H_0+tV) &= \sum_{m=0}^M \sum_{|k|=0}^K \frac{C_m(k) (-1)^{m+|k|} t^m}{(m+|k|)!} V^{(k_1)} \cdots V^{(k_m)} \frac{\d^{m+|k|}\varphi_{\alpha,s}}{\d z^{m+|k|}}  \vert_{z = H_0}  + E(M,K,t,\alpha,s).
\end{align*}
Choose $M = \lfloor n\rfloor$ and for $0 \leq m \leq M$ let $K = M-m$. Since $V^{(k_1)} \cdots V^{(k_m)}$ is a differential operator of order $|k|$, we can write 
\begin{align*}
V^{(k_1)} \cdots V^{(k_m)} = \sum_{|\beta| = 0}^{|k|} g_{k,\beta} \partial^{\beta},
\end{align*}
where $\beta$ is a multi-index of length $n$ and $g_{k,\beta} \in C_c^\infty(\R^n)$. Then we can use Lemma \ref{lem:folland} to compute
\begin{align*}
&\textup{Tr} \left( V V^{(k_1)} \cdots V^{(k_m)} \frac{\d^{m+|k|}\varphi_{\alpha,s}}{\d z^{m+|k|}}  \vert_{z = H_0}  \right) = \sum_{|\beta| = 0}^{|k|} \textup{Tr} \left( V g_{k,\beta} \partial^{\beta} \frac{\d^{m+|k|}\varphi_{\alpha,s}}{\d z^{m+|k|}}  \vert_{z = H_0}  \right) \\
&= (2\pi)^{-n} \sum_{|\beta|=0}^{|k|} \left( \int_{\R^n} V(x) g_{k,\beta}(x) \, \d x \right) \left( \int_{\R^n} (-i)^{|\beta|} \xi^\beta \frac{\d^{m+|k|}\varphi_{\alpha,s}}{\d z^{m+|k|}} (|\xi|^2) \, \d \xi  \right) \\
&=  \mathop{\sum_{|\beta|=0}^{|k|}}_{\beta \textup{ even}}\! \frac{2(-i)^{|\beta|} \Gamma\left( \frac{\beta_1+1}{2} \right) \cdots \Gamma \left( \frac{\beta_n+1}{2} \right)}{(2\pi)^n\Gamma \left( \frac{n+|\beta|}{2} \right)} \left( \int_{\R^n} \! V(x) g_{k,\beta}(x) \! \d x \right) \left( \int_0^\infty \! r^{n+|\beta|-1} \frac{\d^{m+|k|} \varphi_{\alpha,s}}{\d z^{m+|k|}} (r^2) \, \d r \right) \\
&= \mathop{\sum_{|\beta|=0}^{|k|}}_{\beta \textup{ even}}\! \frac{(-i)^{|\beta|} \Gamma\left( \frac{\beta_1+1}{2} \right) \cdots \Gamma \left( \frac{\beta_n+1}{2} \right)}{(2\pi)^n \Gamma \left( \frac{n+|\beta|}{2} \right)} \left( \int_{\R^n} \! V(x) g_{k,\beta}(x) \, \d x \right) \left( \int_0^\infty \!u^{\frac{n+|\beta|}{2}-1} \frac{\d^{m+|k|} \varphi_{\alpha,s}}{\d u^{m+|k|}} (u) \, \d u \right)\!,
\end{align*}
where the sum is over multi-indices $\beta$ with all $\beta_j$ even.
First, suppose that $m+|k| \leq \frac{n+|\beta|}{2}-1$. Integrating by parts in the $u$ integral $(m+|k|)$ times we find
\begin{align*}
&\textup{Tr} \left( V V^{(k_1)} \cdots V^{(k_m)} \frac{\d^{m+|k|}\varphi_{\alpha,s}}{\d z^{m+|k|}}  \vert_{z = H_0}  \right) \\
&=  \mathop{\sum_{|\beta|=0}^{|k|}}_{\beta \textup{ even}} \frac{(-i)^{|\beta|} (-1)^{m+|k|} \Gamma\left( \frac{\beta_1+1}{2} \right) \cdots \Gamma \left( \frac{\beta_n+1}{2} \right)}{(2\pi)^n \Gamma \left( \frac{n+|\beta|}{2} \right)} \left( \int_{\R^n} V(x) g_{k,\beta}(x) \, \d x \right) \\
&\times \left( \int_0^\infty \frac{\d^{m+|k|}}{\d u^{m+|k|}} \left( u^{\frac{n+|\beta|}{2}-1} \right)  \varphi_{\alpha,s}(u) \, \d u \right) \\
&=  \mathop{\sum_{|\beta|=0}^{|k|}}_{\beta \textup{ even}} \frac{(-i)^{|\beta|}(-1)^{m+|k|} \Gamma\left( \frac{\beta_1+1}{2} \right) \cdots \Gamma \left( \frac{\beta_n+1}{2} \right) }{(2\pi)^n \Gamma\left( \frac{n+|\beta|}{2}-m-|k| \right)} \left( \int_{\R^n} V(x) g_{k,\beta}(x) \, \d x \right) \\
&\times \left( \int_0^\infty  u^{\frac{n+|\beta|}{2}-m-|k|-1} \varphi_{\alpha,s}(u) \, \d u \right).
\end{align*}
Note that the boundary terms from the integration by parts vanish since $\frac{n+|\beta|}{2}-m-|k|-1 > 0$. We first consider $n$ odd and make the estimate
\begin{align*}
\left|\int_0^\infty  u^{\frac{n+|\beta|}{2}-m-|k|-1} \varphi_{\alpha,s}(u) \, \d u \right| &\leq \int_0^\infty u^{\frac{n+|\beta|}{2}-m-|k|-1} (1+u^2)^{-\textup{Re}(s)} \, \d u \\
&= \frac12 \int_0^\infty v^{\frac{n+|\beta|}{4}-\frac{m+|k|}{2}-1} (1+v)^{-\textup{Re}(s)} \, \d v \\
&= \frac{\Gamma\left(\frac{n+|\beta|}{4}-\frac{m+|k|}{2} \right) \Gamma\left( \textup{Re}(s)+\frac{m+|k|}{2}-\frac{n+|\beta|}{4} \right)}{2 \Gamma \left( \textup{Re}(s) \right)}.
\end{align*}
Since $n$ is odd, we have
\begin{align*}
\frac12 + \frac{m+|k|}{2} - \frac{n+|\beta|}{4} &\notin \Z
\end{align*}
for all possible $m, k, \beta$ and thus we find
\begin{align*}
\mathop{\textup{Res}}_{s=\frac12} \left(\int_0^\infty  u^{\frac{n+|\beta|}{2}-m-|k|-1} \varphi_{\alpha,s}(u) \, \d u \right) &= 0,
\end{align*}
from which we deduce that
\begin{align*}
\mathop{\textup{Res}}_{s=\frac12} \left( \int_0^1 \textup{Tr}\left(V\left( \textup{Id}+(H_0+tV+\alpha)^2 \right)^{-s} \right) \, \d t \right) &=\mathop{\textup{Res}}_{s=\frac12} \left( \int_0^1 \textup{Tr} \left( VE(M,K,t,\alpha,s \right) \, \d t \right),
\end{align*}
which we show is zero below.
We now consider $n$ even. An application of Lemma \ref{lem:even-one-form-residue} gives
\begin{align*}
&\mathop{\textup{Res}}_{s=\frac12} \left(\int_0^\infty ( u^{\frac{n+|\beta|}{2}-m-|k|-1} \varphi_{\alpha,s}(u) \, \d u \right)  \\
&= \sum_{j=0}^{\frac{n+|\beta|}{2}-m-|k|-1} \binom{\frac{n+|\beta|}{2}-m-|k|-1}{j} \frac{(-1)^{\frac{n+|\beta|}{2}-m-|k|-1} \alpha^{\frac{n+|\beta|}{2}-m-|k|-j-1}\Gamma\left(\frac{j+1}{2} \right)}{4 \Gamma\left( \frac{j}{2}+1 \right) \Gamma\left(\frac12 \right)}.
\end{align*}
Next, we consider the case $m+|k|>\frac{n+|\beta|}{2}-1$. If $n$ is even, we integrate by parts $\frac{n+|\beta|}{2}-1$ times in the $u$-integral to obtain
\begin{align*}
\int_0^\infty u^{\frac{n}{2}-1} \frac{\d^{m+|k|} \varphi_{\alpha,s}}{\d u^{m+|k|}} (u) &= (-1)^{\frac{n+|\beta|}{2}-1} \int_0^\infty \left(\frac{\d^{\frac{n+|\beta|}{2}-1}}{\d u^{\frac{n+|\beta|}{2}-1}} u^{\frac{n+|\beta|}{2}-1} \right) \left( \frac{\d^{m+|k|+1-\frac{n+|\beta|}{2}} \varphi_{\alpha,s}}{\d u^{m+|k|+1-\frac{n+|\beta|}{2}}}(u) \right) \, \d u \\
&= (-1)^{\frac{n+|\beta|}{2}} \Gamma\left(\frac{n+|\beta|}{2} \right) \int_0^\infty \frac{\d^{m+|k|+1-\frac{n+|\beta|}{2}} \varphi_{\alpha,s}}{\d u^{m+|k|+1-\frac{n+|\beta|}{2}}}(u) \, \d u \\
&= (-1)^{\frac{n+|\beta|}{2}+1} \Gamma\left(\frac{n+|\beta|}{2} \right) \frac{\d^{m+|k|-\frac{n+|\beta|}{2}} \varphi_{\alpha,s}}{\d u^{m+|k|+1-\frac{n+|\beta|}{2}}}\Big\vert_{u=0},
\end{align*}
which is holomorphic at $s = \frac12$. If $n$ is odd, then a similar estimate to the case $m+|k| \leq \frac{n+|\beta|}{2}-1$ shows that the contribution is holomorphic at $s = \frac12$.
So we find
\begin{align*}
&\mathop{\textup{Res}}_{s=\frac12} \left(\int_0^1 t^m \textup{Tr} \left( V V^{(k_1)} \cdots V^{(k_m)} \frac{\d^{m+|k|}\varphi_{\alpha,s}}{\d z^{m+|k|}}  \vert_{z = H_0}  \right) \, \d t \right) \\
&=  \mathop{\sum_{|\beta|=0}^{|k|}}_{\beta \textup{ even}} \mathop{\sum_{j=0}^{\frac{n+|\beta|}{2}-m-|k|-1}}_{j \textup{ even}}  \left( \int_{\R^n} V(x) g_{k,\beta}(x) \, \d x \right) \binom{\frac{n+|\beta|}{2}-m-|k|-1}{j} \\
&\times \frac{(-i)^{|\beta|}(-1)^{\frac{n+|\beta|}{2}-m-|k|-\frac{j+1}{2}} \alpha^{\frac{n+|\beta|}{2}-m-|k|-j-1}  \Gamma\left( \frac{\beta_1+1}{2} \right) \cdots \Gamma \left( \frac{\beta_n+1}{2} \right) \Gamma\left(\frac{j+1}{2} \right)}{(2\pi)^n(m+1)\Gamma\left(\frac{j}{2}+1 \right)\Gamma\left(\frac12 \right) \Gamma\left( \frac{n+|\beta|}{2}-m-|k| \right)}.
\end{align*}
For $\ell \in \N$ define the set
\begin{align*}
Q_{M,K}(\ell) &= \{(m,k,\beta) \in \{ 0,1, \dots, M\} \times \{0,1,\dots,K\}^m \times \{0,1, \dots, K\}^n : |\beta| \leq |k|, \\
&\quad m+|k|+1- \frac{|\beta|}{2} = \ell \}.
\end{align*}
Recalling the coefficients $C_\ell(n,V)$ of the high-energy polynomial for $\xi$ of Equation \eqref{eq:high-energy-coefficients} we have that 
\begin{align*}
&\mathop{\textup{Res}}_{s=\frac12} \left( \sum_{m=0}^M \int_0^1 t^m \sum_{|k|=0}^K  \frac{C_m(k) (-1)^{m+|k|}}{(m+|k|)!}\textup{Tr} \left( V V^{(k_1)} \cdots V^{(k_m)} \frac{\d^{m+|k|}\varphi_{\alpha,s}}{\d z^{m+|k|}}  \vert_{z = H_0}  \right) \, \d t \right) \\
&=  \sum_{m=0}^M \sum_{|k|=0}^K \mathop{\sum_{|\beta|=0}^{|k|}}_{\beta \textup{ even}} \mathop{\sum_{j=0}^{\frac{n+|\beta|}{2}-m-|k|-1}}_{j \textup{ even}}  \left( \int_{\R^n} V(x) g_{k,\beta}(x) \, \d x \right) \binom{\frac{n+|\beta|}{2}-m-|k|-1}{j} \\
&\times \frac{(-i)^{|\beta|}(-1)^{\frac{n+|\beta|}{2}-m-|k|-\frac{j+1}{2}} \alpha^{\frac{n+|\beta|}{2}-m-|k|-j-1}  \Gamma\left( \frac{\beta_1+1}{2} \right) \cdots \Gamma \left( \frac{\beta_n+1}{2} \right) \Gamma\left(\frac{j+1}{2} \right)}{(2\pi)^n(m+1)\Gamma\left(\frac{j}{2}+1 \right)\Gamma\left(\frac12 \right) \Gamma\left( \frac{n+|\beta|}{2}-m-|k| \right)}\\
&= \sum_{\ell = 1}^{\frac{n}{2}} \mathop{\sum_{j=0}^{\frac{n}{2}-\ell}}_{j \textup{ even}} \binom{\frac{n}{2}-\ell}{j} \frac{(-1)^{\frac{n}{2}-\ell-\frac{j}{2}-1} \alpha^{\frac{n}{2}-\ell-j-1}\Gamma\left(\frac{j+1}{2} \right)}{4(2\pi i)\Gamma\left(\frac{j}{2}+1 \right) \Gamma\left(\frac12 \right)} C_\ell(n,V).
\end{align*}
We now consider the contribution from the remainder term $E(M,K,t,\alpha,s)$. There are two types of terms to consider. The first are those involving $P_{m,K,t}(z)$. Since $V \in C_c^\infty(\R^n)$ we can factorise $V = q_1 q_2$ with $q_1, q_2 \in C_c^\infty(\R^n)$. Since $P_{m,K,t}(z)$ has order at most $-2m-K-3$, there exists $C > 0$ such that
\begin{align*}
\norm{R_0(z)^{-m-\frac{K}{2}-\frac{3}{2}}q_2 P_{m,K,t}(z) } &\leq C.
\end{align*}
Note also that
\begin{align*}
\norm{q_1 R_0(z)^{m+\frac{K}{2}+\frac{3}{2}}}_1 &\leq (a^2+v^2)^{-\frac{m}{2}-\frac{K}{4}-\frac{3}{4}+\frac{n+1}{4}},
\end{align*}
which follows from a careful application of the Rellich lemma. Combining these we make the estimate
\begin{align*}
\norm{\int_\gamma \varphi_{\alpha,s}(z) P_{m,K,t}(z) \, \d z }_1 &\leq C \int_\R (1+((a+\alpha)^2+v^2))^{-\frac{\textup{Re}(s)}{2}} (a^2+v^2)^{-\frac{m}{2}-\frac{K}{4}-\frac{3}{4}+\frac{n+1}{4}} \, \d v,
\end{align*}
which is finite for $\textup{Re}(s)+m+\frac{K}{2}+\frac{3}{2} > \frac{n+1}{2}+1$. Recalling that $\textup{Re}(s) > \frac12$ and we have chosen $M =n$ and $K = M-m$ guarantees convergence. A similar argument to \cite[Lemma 7.4]{CPRS} shows that this contribution is holomorphic at $s = \frac12$, as is the contribution from the terms containing $R_t(z)$.
Thus for both $n$ even and odd, we find 
\begin{align*}
\mathop{\textup{Res}}_{s=\frac12} \left( \int_0^1 \textup{Tr} \left( VE(M,K,t,\alpha,s \right) \, \d t \right) &= 0,
\end{align*}
which completes the proof.
\end{proof}

\subsection{The Birman-Kre\u{\i}n term}\label{sec:eta-term}

In this subsection we use the Birman-Kre\u{\i}n trace formula to determine the kernel and $\eta$ contributions to the spectral flow.
\begin{lemma}
By construction, the projections $P_{\textup{Ker}(H(\alpha))}$ and $P_{\textup{Ker}(H_0(\alpha))}$ are both zero.
\end{lemma}

Since the kernel terms both vanish we are now able to evaluate the $\eta$ contributions. We note that by Proposition \ref{prop:one-form-res} the residue of the integral of one form contribution to Equation \eqref{eq:spec-flow-residue} at $s = \frac12$ exists, and thus so does the residue of the Birman-Kre\u{\i}n contribution at $s=\frac12$.
\begin{lemma}\label{lem:kernel-vanish}
Suppose that $V \in C_c^\infty(\R^n)$. Then the $\eta$ contribution to the Hamiltonian spectral flow is given by
\begin{align*}
&\Res{s=\frac12} \Big( C_s \big(\textup{Tr}( \eta_s(H(\alpha))-\eta_s(H_0(\alpha)) \big) \Big) \\
&= N+N_{res} + \Res{s=\frac12} \left(\frac{1}{2\pi i} \int_\alpha^\infty C_s \eta_s(\lambda) \textup{Tr}(S_\alpha(\lambda)^*S_\alpha'(\lambda)) \, \d \lambda \right),
\end{align*}
where $N$ is the number of eigenvalues of $H =H_0+V$, counted with multiplicity,
%. Here $\xi$ refers to the spectral shift function for the pair $(H,H_0)$ 
and $N_{res}$ is the contribution from resonances as defined in Theorem \ref{thm:spec-shift-values}.
\end{lemma}
\begin{proof}
Choose $s > \frac{n}{2}+1$, so that Equation \eqref{eq:ssf-condition} is satisfied for $\eta_s$ and thus the Birman-Kre\u{\i}n trace formula can be applied to $\eta_s$ . Enumerate the distinct eigenvalues of $H(\alpha)$ as $0 < \lambda_1(\alpha) < \cdots < \lambda_K(\alpha) \leq \alpha$. We prove the result in the case $\lambda_K(\alpha) = \alpha$, that is in  the case that zero is an eigenvalue for $H = H_0+V$. Apply the Birman-Kre\u{i}n trace formula to obtain
\begin{align*}
\textup{Tr} \left(\eta_s(H(\alpha)) - \eta_s(H_0(\alpha)) \right) &= \frac{1}{2\pi i } \int_\alpha^\infty\eta_s(\lambda) \textup{Tr}(S_\alpha(\lambda)^*S_\alpha'(\lambda)) \, \d \lambda +\sum_{k=1}^{K-1} M(\lambda_k(\alpha)) \eta_s(\lambda_k(\alpha)) \\
&+M(\alpha)\eta_s(\alpha)+\eta_s(\alpha) \left(\xi_\alpha(\alpha-)-\xi_\alpha(\alpha+)-M(\alpha) \right), 
\end{align*}
where $\xi_\alpha$ is the spectral shift function for the pair $(H(\alpha), H_0(\alpha))$ and $M(\lambda_j(\alpha))$ denotes the multiplicity of the eigenvalue $\lambda_j(\alpha)$ for the operator $H(\alpha)$. Recall that by construction, we have $\xi_\alpha(\lambda) = \xi(\lambda-\alpha)$ so that $\xi_\alpha(\alpha\pm)=\xi(0\pm)$.
% and $M(\lambda_j(\alpha))$ denotes the multiplicity of the $j$-th distinct eigenvalue $\lambda_j$ of $H = H_0+V$. 
Thus after multiplying by $C_s$ we have
\begin{align}\label{eq:eta-trace-mero}
C_s\textup{Tr} \left(\eta_s(H(\alpha)) - \eta_s(H_0(\alpha)) \right) &= \frac{C_s}{2\pi i } \int_\alpha^\infty\eta_s(\lambda) \textup{Tr}(S_\alpha(\lambda)^*S_\alpha'(\lambda)) \, \d \lambda \\
&+C_s\sum_{k=1}^{K} M(\lambda_k(\alpha)) \eta_s(\lambda_k(\alpha))+C_s \eta_s(\alpha) \left(\xi(0-)-\xi(0+)-N_0 \right). \nonumber
\end{align}
Observe that for $x \neq 0$ we have $\displaystyle \Res{s=\frac12} \left( C_s \eta_s(x) \right) = \textup{sign}(x)$.

The left hand side of Equation \eqref{eq:eta-trace-mero} has a residue at $s=\frac12$ if and only if the first term on the right hand side does. 
%defines a meromorphic continuation of the left-hand side, with a simple pole at $s = \frac12$.  
Note that by Theorem \ref{thm:spec-shift-values} we have $\xi(0-)-\xi(0+)-N_0 = N_{res}$. It remains to take the residue at $s = \frac12$. 

By construction we have $\lambda_j(\alpha) > 0$ for all $j$ and thus
\begin{align*}
&\Res{s=\frac12} \left( C_s \textup{Tr}\left( \eta_s(H(\alpha))-\eta_s(H_0(\alpha)) \right) \right) \\
&= \Res{s=\frac12} \Bigg(\frac{C_s}{2\pi i } \int_\alpha^\infty\eta_s(\lambda) \textup{Tr}(S_\alpha(\lambda)^*S_\alpha'(\lambda)) \, \d \lambda +\sum_{k=1}^{K} M(\lambda_k) C_s\eta_s(\lambda_k(\alpha)) +C_s \eta_s(\alpha) N_{res} \Bigg) \\
&= N+N_{res} + \Res{s=\frac12} \left(\frac{1}{2\pi i} \int_\alpha^\infty C_s \eta_s(\lambda) \textup{Tr}(S_\alpha(\lambda)^*S_\alpha'(\lambda)) \, \d \lambda \right),
\end{align*}
as claimed.
\end{proof}
We can now compute the residue of the Birman-Kre\u{\i}n integral contribution to the spectral flow with the aid of a technical result.

\begin{lemma}\label{lem:residue-dominated}
Fix $\beta \geq 0$ and $f, g: \R^+ \to \C$ with $f-g \in L^1(\R^+)$. Suppose in addition that
\begin{align*}
&\Res{s=\frac12} \left( C_s \int_\beta^\infty \eta_s(\lambda) f(\lambda-\beta) \, \d \lambda \right)
\end{align*}
exists. Then 
\begin{align*}
\Res{s=\frac12} \left( C_s \int_\beta^\infty \eta_s(\lambda) f(\lambda-\beta) \, \d \lambda \right) &= \int_0^\infty (f(\lambda)-g(\lambda)) \, \d \lambda+\Res{s=\frac12} \left( C_s \int_\beta^\infty \eta_s(\lambda) g(\lambda-\beta) \, \d \lambda \right).
\end{align*}
\end{lemma}
\begin{proof}
Adding zero gives
\begin{align}\label{eq:residue-lhs-rhs}
&\Res{s=\frac12} \left( C_s \int_\beta^\infty \eta_s(\lambda) f(\lambda-\beta) \, \d \lambda \right) \\
&= \Res{s=\frac12} \left( C_s \int_\beta^\infty \eta_s(\lambda) (f(\lambda-\beta)-g(\lambda-\beta)) \, \d \lambda + C_s \int_\beta^\infty \eta_s(\lambda) g(\lambda-\beta) \, \d \lambda \right). \nonumber
\end{align}

%
%Recall that $C_s$ has a simple pole at $s =\frac12$ with residue $1$. Fix $\eps > 0$. Since the residue at $s = \frac12$ of $C_s$ exists, there exists $\delta > 0$ such that $|s-\frac12| < \delta$ gives $|(s-\frac12)C_s - 1 | < \eps$. Then for $|s-\frac12| < \delta$ and $x > 0$ we have
%\begin{align*}
%\left|\left(s-\frac12 \right)C_s \eta_s(x)-1 \right| &= \left|(2s-1) \int_x^\infty (1+v^2)^{-s} \d v - 1 \right| \\
%&\leq \left| (2s-1) \int_0^\infty (1+v^2)^{-s} \, \d v - 1 \right| + |2s-1| \int_0^x (1+v^2)^{-\textup{Re}(s)} \, \d v \\
%&\leq \eps +|2s-1| x(1+x^2)^{-\textup{Re}(s)} \leq \eps+2\delta.
%\end{align*}
%Thus for sufficiently small $|s-\frac12|$ we have $(s-\frac12)C_s \eta_s(x)$ is bounded independently of $s$ and $x$. 

One straightforwardly checks that 
$\Res{s=1/2}\eta_s(x)=1$ for all $x>0$.
Then an application of the dominated convergence theorem allows us to compute that
\begin{align*}
\Res{s=\frac12} \left(\int_\beta^\infty \eta_s(\lambda) (f(\lambda-\beta)-g(\lambda-\beta)) \, \d \lambda \right) &= \int_\beta^\infty \Res{s=\frac12} \eta_s(\lambda) (f(\lambda-\beta)-g(\lambda-\beta)) \, \d \lambda \\
&= \int_\beta^\infty (f(\lambda-\beta)-g(\lambda-\beta) ) \, \d \lambda.
\end{align*}
Since the residue of the first term on the right-hand side of Equation \eqref{eq:residue-lhs-rhs} exists, so does the residue of the second term on the right-hand side. Making the substitution $u = \lambda-\beta$ completes the proof.
\end{proof}
\begin{prop}\label{prop:birman-krein-res}
Suppose that $V \in C_c^\infty(\R^n)$ and $\alpha > -2\nu$. Then for $n$ odd we have
\begin{align*}
 \Res{s=\frac12} \left(\frac{C_s}{2\pi i} \int_\alpha^\infty  \textup{Tr}(S_\alpha(\lambda)^*S_\alpha'(\lambda)) \eta_s(\lambda) \, \d \lambda \right) &= \frac{1}{2\pi i} \int_0^\infty \left( \textup{Tr}(S(\lambda)^*S'(\lambda)) - p_n(\lambda) \right) \, \d \lambda.
\end{align*}
If $n$ is even we have
\begin{align*}
\Res{s=\frac12} \Big(\frac{C_s}{2\pi i} \int_\alpha^\infty \textup{Tr}(S_\alpha(\lambda)^*S_\alpha'(\lambda)) &\eta_s(\lambda) \, \d \lambda \Big) 
= \frac{1}{2\pi i} \int_0^\infty \left( \textup{Tr}(S(\lambda)^*S'(\lambda)) - p_n(\lambda) \right) \, \d \lambda \\
&\quad + \sum_{\ell=1}^{\lfloor \frac{n-1}{2} \rfloor} \mathop{\sum_{j=0}^{\frac{n}{2}-\ell}}_{j \textup{ even }} \binom{\frac{n}{2}-\ell}{j} \frac{(-1)^{\frac{n}{2}-\ell -\frac{j}{2}} \alpha^{\frac{n}{2}-\ell-j} \Gamma\left(\frac{j+1}{2} \right)}{2 (2\pi i) \Gamma\left(\frac{j}{2}+1 \right) \Gamma\left( \frac12 \right)} C_\ell(n,V),
\end{align*}
where the $C_\ell(n,V)$ are the high-energy coefficients for $P_n$ defined in Equation \eqref{eq:high-energy-coefficients}
\end{prop}
\begin{proof}
%For $\textup{Re}(s)$ large, write 
%\begin{align*}
%\frac{C_s}{2\pi i} \int_\alpha^\infty \textup{Tr}(S_\alpha(\lambda)^*S_\alpha'(\lambda)) \eta_s(\lambda) \, \d \lambda &= \frac{C_s}{2\pi i} \int_\alpha^\infty \left( \textup{Tr}(S_\alpha(\lambda)^*S_\alpha'(\lambda)) -p_n(\lambda-\alpha) \right)\eta_s(\lambda) \, \d \lambda\\
%&\quad + \frac{C_s}{2\pi i} \int_\alpha^\infty p_n(\lambda-\alpha) \eta_s(\lambda) \, \d \lambda
%\end{align*}
Note that by Lemma \ref{lem:ssf-L1} the map $[0,\infty) \ni \lambda \mapsto \textup{Tr}(S(\lambda)^*S'(\lambda)) -p_n(\lambda)$ is integrable on $[0,\infty)$ and thus we can apply Lemma \ref{lem:residue-dominated} to obtain
\begin{align*}
&\Res{s=\frac12} \left( \frac{C_s}{2\pi i} \int_\alpha^\infty \left( \textup{Tr}(S_\alpha(\lambda)^*S_\alpha'(\lambda)) \right) \eta_s(\lambda) \, \d \lambda \right) \\
&=\frac{1}{2 \pi i} \int_0^\infty \left( \textup{Tr}(S(\lambda)^*S'(\lambda)) - p_n(\lambda) \right) \, \d \lambda + \frac{1}{2\pi i}\Res{s=\frac12} \left( C_s \int_\alpha^\infty p_n(\lambda-\alpha) \eta_s(\lambda) \, \d \lambda \right).
\end{align*}
Thus it remains to compute
\begin{align*}
\Res{s=\frac12} \left(\frac{C_s}{2\pi i} \int_\alpha^\infty p_n(\lambda-\alpha) \eta_s(\lambda) \, \d \lambda \right) &= \Res{s=\frac12} \left( \sum_{\ell=1}^{\lfloor \frac{n-1}{2} \rfloor} \frac{c_\ell(n,V)}{2\pi i}  C_s \int_\alpha^\infty (\lambda-\alpha)^{\frac{n}{2}-\ell-1} \eta_s(\lambda) \, \d \lambda \right) \\
&=  \Res{s=\frac12} \left( \sum_{\ell=1}^{\lfloor \frac{n-1}{2} \rfloor} \frac{c_\ell(n,V)}{2\pi i} C_s \int_0^\infty u^{\frac{n}{2}-\ell-1} \eta_s(u+\alpha) \, \d u \right).
\end{align*}
We show that the residue of each of the terms in the sum exist individually, so that the summation can be passed through the residue. First we consider $n$ odd. We integrate by parts to find
\begin{align*}
 C_s \int_0^\infty &u^{\frac{n}{2}-\ell-1} \eta_s(u+\alpha) \, \d u   
 =  \int_0^\infty \int_{u+\alpha}^\infty u^{\frac{n}{2}-\ell-1} (1+v^2)^{-s} \, \d v \, \d u  \\
%&\leq \int_0^\infty \int_{u+\alpha}^\infty u^{\frac{n}{2}-\ell-1} (1+v^2)^{-\textup{Re}(s)} \, \d v \, \d u \\
&= \left[ \frac{u^{\frac{n}{2}-\ell}}{\frac{n}{2}-\ell} \int_{u+\alpha}^\infty (1+v^2)^{-s} \, \d v \right]_0^\infty   +  \frac{1}{\frac{n}{2}-\ell} \int_0^\infty u^{\frac{n}{2}-\ell} (1+(u+\alpha)^2)^{-s} \, \d u \\
&=\frac{1}{\frac{n}{2}-\ell} \int_0^\infty u^{\frac{n}{2}-\ell} (1+u^2)^{-s} \, \d u
+\frac{1}{\frac{n}{2}-\ell} \int_0^\infty u^{\frac{n}{2}-\ell} \big((1+(u+\alpha)^2)^{-s}-(1+u^2)^{-s} \, \d u\\
&=\frac{\Gamma\left(\frac{n}{4}-\frac{\ell}{2}+\frac12\right)\Gamma\left(s-\frac{n}{4}+\frac{\ell}{2}-\frac12 \right)}{2\Gamma(s)}+holo(s)
\end{align*}
where $holo$ is a function holomorphic at $s=\frac12$.
Since $n$ is odd, we have $1 - \frac{n}{4} + \frac{\ell}{2} \notin \Z$ and thus 
\begin{align*}
\Res{s=\frac12} \left(\frac{C_s}{2\pi i} \int_\alpha^\infty p_n(\lambda-\alpha) \eta_s(\lambda) \, \d \lambda \right) &= \sum_{\ell=1}^{\lfloor \frac{n-1}{2} \rfloor} \frac{c_\ell(n,V)}{2\pi i} \Res{s=\frac12} \left( C_s \int_0^\infty u^{\frac{n}{2}-\ell-1} \eta_s(u+\alpha) \, \d u \right) \\
&= 0.
\end{align*}
Now we consider $n$ even. In this case we integrate by parts to write
\begin{align*}
\int_\alpha^\infty (\lambda-\alpha)^{\frac{n}{2}-\ell-1} \eta_s(\lambda) \, \d \lambda &= -\frac{1}{\frac{n}{2}-\ell} \int_\alpha^\infty (\lambda-\alpha)^{\frac{n}{2}-\ell} \frac{\d}{\d \lambda} \eta_s(\lambda) \, \d \lambda \\
&= \frac{1}{\frac{n}{2}-\ell} \int_\alpha^\infty (\lambda-\alpha)^{\frac{n}{2}-\ell} (1+\lambda^2)^{-s} \, \d \lambda.
\end{align*}
We now use the binomial expansion to obtain
\begin{align*}
\int_\alpha^\infty (\lambda-\alpha)^{\frac{n}{2}-\ell-1} \eta_s(\lambda) \, \d \lambda &= \frac{1}{\frac{n}{2}-\ell} \sum_{j=0}^{\frac{n}{2}-\ell} \binom{\frac{n}{2}-\ell}{j}  (-\alpha)^{\frac{n}{2}-\ell-j} \int_\alpha^\infty \lambda^j (1+\lambda^2)^{-s}\, \d \lambda.
\end{align*}
Returning to the residue calculation we find
\begin{align*}
&\Res{s=\frac12} \left( C_s \int_\alpha^\infty (\lambda-\alpha)^{\frac{n}{2}-\ell-1} \eta_s(\lambda) \, \d \lambda \right) \\
&= \frac{1}{\frac{n}{2}-\ell} \sum_{j=0}^{\frac{n}{2}-\ell} \binom{\frac{n}{2}-\ell}{j} (-\alpha)^{\frac{n}{2}-\ell-j}\Res{s=\frac12} \left(\int_\alpha^\infty \lambda^j (1+\lambda^2)^{-s}\, \d \lambda \right) \\
&= \frac{1}{\frac{n}{2}-\ell} \sum_{j=0}^{\frac{n}{2}-\ell} \binom{\frac{n}{2}-\ell}{j} (-\alpha)^{\frac{n}{2}-\ell-j}\Res{s=\frac12} \left(\int_0^\infty \lambda^j (1+\lambda^2)^{-s}\, \d \lambda \right) \\
&= \frac{1}{\frac{n}{2}-\ell} \sum_{j=0}^{\frac{n}{2}-\ell} \binom{\frac{n}{2}-\ell}{j} (-\alpha)^{\frac{n}{2}-\ell-j}\Res{s=\frac12} \left(\frac{\Gamma\left( \frac{j+1}{2} \right) \Gamma\left(s-\frac{j+1}{2} \right)}{2\Gamma(s)} \right) \\
&= \frac{1}{\frac{n}{2}-\ell} \mathop{\sum_{j=0}^{\frac{n}{2}-\ell}}_{j \textup{ even }} \binom{\frac{n}{2}-\ell}{j} \frac{(-1)^{\frac{n}{2}-\ell -\frac{j}{2}} \alpha^{\frac{n}{2}-\ell-j} \Gamma\left(\frac{j+1}{2} \right)}{2 \Gamma\left(\frac{j}{2}+1 \right) \Gamma\left( \frac12 \right)},
\end{align*}
from which the statement follows by observing the relation $c_\ell(n,V) = \left(\frac{n}{2}-\ell \right) C_\ell(n,V)$.
\end{proof}

%%%%%%%%%%%%%%%%%%%%%%%%%%%%%%%%%%%%%%%%%%%%%%%%%%%%%%%%%%%%%%%%%%%%%%%%%%%
%%%%%%%%%%%%%%%%%%%%%%%%%%%%%%%%%%%%%%%%%%%%%%%%%%%%%%%%%%%%%%%%%%%%%%%%%%%
%%%%%%%%%%%%%%%%%%%%%%%%%%%%%%%%%%%%%%%%%%%%%%%%%%%%%%%%%%%%%%%%%%%%%%%%%%%
%%%%%%%%%%%%%%%%%%%%%%%%%%%%%%%%%%%%%%%%%%%%%%%%%%%%%%%%%%%%%%%%%%%%%%%%%%%

\section{The spectral flow formula and Levinson's theorem}\label{sec:levinson}

In this section we return to the spectral flow formula of Equation \eqref{eq:spec-flow-residue} applied to the path $H_t(\alpha)$ and, using the results of Section \ref{sec:spec-flow} we can prove Levinson's theorem in all dimensions. 

\begin{thm}\label{thm:spec-flow-final}
Suppose that $V \in C_c^\infty(\R^n)$ and $\alpha > -2\nu$. Then the spectral flow along the path $H_t(\alpha)$ is given by
\begin{align*}
\textup{sf}(H_t(\alpha)) &= \frac12 (N+N_{res})+ \frac{1}{4\pi i} \int_0^\infty \left(\textup{Tr}(S(\lambda)^*S'(\lambda)) - p_n(\lambda) \right) \, \d \lambda -\frac12 \beta_n(V).
\end{align*}
\end{thm}
\begin{proof}
Lemma \ref{lem:spec-flow-residue} and Lemma \ref{lem:kernel-vanish} give that 
\begin{align*}
\textup{sf}(H_t(\alpha)) \nonumber &= \Res{s=\frac12} \left( C_s \int_0^1 \Tr \left( V (\textup{Id}+H_t(\alpha)^2)^{-s} \right) \d t + \frac{C_s}{2} \Tr \big( \eta_s(H(\alpha))-\eta_s(H_0(\alpha)) \big) \right).
\end{align*}
Suppose first that $n$ is odd. Applying Proposition \ref{prop:one-form-res} to the first term on the right-hand side and Proposition \ref{prop:birman-krein-res} to the second term gives
\begin{align*}
\textup{sf}(H_t(\alpha)) &= \frac12 (N+N_{res}) + \frac{1}{4\pi i} \int_0^\infty \left( \textup{Tr}(S(\lambda)^*S'(\lambda)) - p_n(\lambda) \right) \, \d \lambda.
\end{align*}
Now we consider $n$ even. Applying again Propositions \ref{prop:one-form-res} and \ref{prop:birman-krein-res} gives
\begin{align*}
\textup{sf}(H_t(\alpha)) &= \frac12 (N+N_{res}) + \frac{1}{4\pi i} \int_0^\infty \left( \textup{Tr}(S(\lambda)^*S'(\lambda)) - p_n(\lambda) \right) \, \d \lambda \\
&+ \frac12 \sum_{\ell=1}^{\lfloor \frac{n-1}{2} \rfloor} \mathop{\sum_{j=0}^{\frac{n}{2}-\ell}}_{j \textup{ even }} \binom{\frac{n}{2}-\ell}{j} \frac{(-1)^{\frac{n}{2}-\ell -\frac{j}{2}} \alpha^{\frac{n}{2}-\ell-j} \Gamma\left(\frac{j+1}{2} \right)}{2 (2\pi i) \Gamma\left(\frac{j}{2}+1 \right) \Gamma\left( \frac12 \right)} C_\ell(n,V)\\
&+ \sum_{\ell = 1}^{\frac{n}{2}} \mathop{\sum_{j=0}^{\frac{n}{2}-\ell}}_{j \textup{ even}} \binom{\frac{n}{2}-\ell}{j} \frac{(-1)^{\frac{n}{2}-\ell-\frac{j}{2}-1} \alpha^{\frac{n}{2}-\ell-j}\Gamma\left(\frac{j+1}{2}\right)}{4(2\pi i)\Gamma\left(\frac{j}{2}+1\right)\Gamma\left(\frac12 \right)} C_\ell(n,V).
\end{align*}
It remains to observe that for $n$ even we have $\lfloor \frac{n-1}{2} \rfloor = \frac{n}{2} - 1$ and thus
\begin{align*}
&\sum_{\ell = 1}^{\frac{n}{2}} \mathop{\sum_{j=0}^{\frac{n}{2}-\ell}}_{j \textup{ even}} \binom{\frac{n}{2}-\ell}{j} \frac{(-1)^{\frac{n}{2}-\ell-\frac{j}{2}-1} \alpha^{\frac{n}{2}-\ell-j}\Gamma\left(\frac{j+1}{2}\right)}{4(2\pi i)\Gamma\left(\frac{j}{2}+1\right)\Gamma\left(\frac12 \right)} C_\ell(n,V) \\
&= - \frac12 \sum_{\ell=1}^{\lfloor \frac{n-1}{2} \rfloor} \mathop{\sum_{j=0}^{\frac{n}{2}-\ell}}_{j \textup{ even }} \binom{\frac{n}{2}-\ell}{j} \frac{(-1)^{\frac{n}{2}-\ell -\frac{j}{2}} \alpha^{\frac{n}{2}-\ell-j} \Gamma\left(\frac{j+1}{2} \right)}{2 (2\pi i) \Gamma\left(\frac{j}{2}+1 \right) \Gamma\left( \frac12 \right)} C_\ell(n,V) -\frac12 \beta_n(V),
\end{align*}
where we have used the definition of $\beta_n(V)$ in Equation \eqref{eq:beta-defn}.
\end{proof}
We are now able to prove Levinson's theorem as a consequence of spectral flow along the path $H_t(\alpha)$.

\begin{thm}[Levinson's theorem]\label{thm:levinson}
Suppose that $V \in C_c^\infty(\R^n)$. Then the number $N$ of eigenvalues (counted with multiplicity) of $H = H_0+V$ is given by
\begin{align*}
-N &= \frac{1}{2\pi i}\int_0^\infty \left(\textup{Tr}(S(\lambda)^*S'(\lambda)) - p_n(\lambda) \right) \, \d \lambda - \beta_n(V) +N_{res},
\end{align*}
where $N_{res}$ is as defined in Theorem \ref{thm:spec-shift-values}.
\end{thm}
\begin{proof}
By construction we know that for $\alpha > -2\nu$ we have
\begin{align}\label{eq:spec-flow-value}
\textup{sf}(H_t(\alpha)) &= 0,
\end{align}
since there is no spectrum which moves through zero from right to left as the path is traversed from $H_0(\alpha)$ to $H(\alpha)$. Substituting Equation \eqref{eq:spec-flow-value} into the result of Theorem \ref{thm:spec-flow-final} and solving for $N$ completes the proof.
\end{proof}

%%%%%%%%%%%%%%%%%%%%%%%%%%%%%%%%%%%%%%%%%%%%%%%%%%%%%%%%%%%%%%%%%%%%%
%%%%%%%%%%%%%%%%%%%%%%%%%%%%%%%%%%%%%%%%%%%%%%%%%%%%%%%%%%%%%%%%%%%%%
%%%%%%%%%%%%%%%%%%%%%%%%%%%%%%%%%%%%%%%%%%%%%%%%%%%%%%%%%%%%%%%%%%%%%
%
%\appendix


\begin{thebibliography}{99}

%\bibitem{abramowitz72} M. Abramowitz, I. Stegun. {\em Handbook of Mathematical Functions: with Formulas, Graphs, and Mathematical Tables}. Applied Mathematics Series 55. NBS, 10 edition, 1972.

%\bibitem{agmon75}
%S. Agmon. {\em Spectral properties of {S}chr{\"{o}}dinger operators and scattering theory}, Accademia Nazionale dei Lincei; Scuola Normale Superiore di Pisa, Pisa, 1975.

\bibitem{alexander24}
A. Alexander.
{\em Trace formula and Levinson's theorem as an index pairing in the presence of resonances},
arXiv preprint: \href{https://arxiv.org/abs/2402.15979}{https://arxiv.org/abs/2402.15979}, 2024.

\bibitem{alexander23thesis}
A. Alexander.
{\em Topological {L}evinson's theorem via index pairings and spectral flow}, PhD thesis, University of Wollongong, 2024.

\bibitem{ANRR}
A. Alexander, D. T. Nguyen, A. Rennie, S. Richard.
{\em Levinson's theorem for two-dimensional scattering systems: it was a surprise, it is now topological!}, to appear in Journal of Spectral Theory,
arXiv preprint: \href{https://arxiv.org/abs/2311.09650}{https://arxiv.org/abs/2311.09650}, 2023. 

\bibitem{AR23}
A. Alexander, A. Rennie.
{\em Levinson's theorem as an index pairing}, J. Funct. Anal., {\bf 286} (5), 2024.

\bibitem{AR23-4D}
A. Alexander, A. Rennie.
{\em The structure of the wave operator in four dimensions in the presence of resonances}, arXiv preprint: \href{https://arxiv.org/abs/2311.16438}{https://arxiv.org/abs/2311.16438}, 2023.

%\bibitem{amrein96} W. Amrein, A. Boutet de Monvel, V. Georgescu.
%{\em $C_0$-Groups, Commutator Methods and Spectral Theory of $N$-Body Hamiltonians}, Modern Birkh\"{a}user Classics. Birkh\"{a}user, 1st edition, 1996.

%\bibitem{andrews99} G. Andrews, R. Askey, R. Roy. {\em Special functions}, Encyclopedia of mathematics and its applications 71. Cambridge University Press, 1999.

%\bibitem{alsholm71} P. Alsholm, G. Schmidt. {\em Spectral and scattering theory for {S}chr\"{o}dinger operators}.  Arch. Rational Mech. Anal., {\bf 40}, 1971, 281--311.

\bibitem{aps75}
M. F. Atiyah, V. K. Patodi, I. M. Singer.
{\em Spectral asymmetry and {R}iemannian geometry. {I}},
Math. Proc. Cambridge Philos. Soc., {\bf 77}, 1975, 45--69.

\bibitem{aps76}
M. F. Atiyah, V. K. Patodi, I. M. Singer.
{\em Spectral asymmetry and {R}iemannian geometry. {III}},
Math. Proc. Cambridge Philos. Soc., {\bf 79}, 1976, 71--99.

\bibitem{ACDS}
N. Azamov, A. Carey, P. Dodds, F. Sukochev.
{\em Operator integrals, spectral shift, and spectral flow},
Canad. J. Math., {\bf 61} (2), 2009, 241--263.

\bibitem{ACS}
N. Azamov, A. Carey, F. Sukochev.
{\em The spectral shift function and spectral flow},
Comm. Math. Phys., {\bf 276} (1), 2007, 51--91.

\bibitem{banuelos96}
R. Ba\~{n}uelos, A. S\'{a} Barreto.
{\em On the heat trace of {S}chr\"{o}dinger operators},
Comm. Partial Differential Equations, {\bf 20} (11-12), 1995, 2153--2164.



%
%\bibitem{bellissard12}
%J. Bellissard, H. Schulz-Baldes. {\em Scattering theory for lattice operators in dimension {$d \geq 3$}}, Rev. Math. Phys., {\bf 24} (8), 2012, 12500020.

\bibitem{BCPRSW}
M. T. Benameur, A. Carey, J. Phillips, A. Rennie, F. Sukochev, K. Wojciechowski.
{\em An analytic  approach to spectral flow in von {N}eumann algebras}, 
in `Analysis, geometry and topology of elliptic operators', World Sci. Publ., 2006, 297--352.

\bibitem{birman62}
M. \v{S}. Birman, M. G. Kre\u{\i}n.
{\em On the theory of wave operators and scattering operators}, Dokl. Akad. Nauk SSSR, {\bf 144}, 1962, 475--478.

%\bibitem{bolle77}
%D. Boll\'{e}, T. A. Osborn. {\em An extended {L}evinson's theorem}, J. Mathematical Phys., {\bf 18} (3), 1977, 432--440.
%
%\bibitem{bolle85} D. Boll\'{e}, F. Gesztesy, S. F. J. Wilk. {\em A complete treatment of low-energy scattering in one dimensions}, J. Operator Theory, 1985.

%\bibitem{bolle86}
%D. Boll\'{e}, F. Gesztesy, C. Danneels, S. F. J. Wilk,
%{\em Threshold behaviour and Levinson's theorem for two-dimensional scattering systems: a surprise}, Phys. Rev. Lett. {\bf 56}, 1986, 900--903.
%
%\bibitem{bolle88} D. Boll\'{e}, C. Danneels, F. Gesztesy. {\em Threshold scattering in two dimensions}, Ann. Inst. Henri Poincar\'{e} Phys. Th\'{e}or., {\bf 48} (2), 1988.
%
%\bibitem{bruning82}
%E. Br\"{u}ning, F. Gesztesy. {\em Continuity of wave and scattering operators with respect to interactions}, J. Math. Phys., {\bf 24} (6), 1983, 1516--1528.

%\bibitem{CGRS2} A. Carey, V. Gayral, A. Rennie, F. Sukochev, {\em Index theory for locally compact noncommutative spaces}, 
%Mem. Amer. Math. Soc., {\bf 231} (1085), 2014.

\bibitem{CP1}
A. Carey, J. Phillips.
{\em Unbounded {F}redholm modules and spectral flow}, 
Canad. J. Math. {\bf 50} (4), 1998, 673--718.

%\bibitem{CPR4} A. Carey, J. Phillips, A. Rennie, {\em Spectral triples: examples and index theory}, 
%in `Noncommutative Geometry and Physics: Renormalisation, Motives, Index Theory', Ed. A. Carey, EMS, 2011, 175--265.

\bibitem{CPRS}
A. Carey, J. Phillips, A. Rennie, F. Sukochev. {\em The local index formula in semifinite von {N}eumann algebras {I}: {S}pectral flow}, Adv. Math., {\bf 202} (2), 2006, 451--516.

\bibitem{CPotSuk}
A. Carey, D. Potatov, F. Sukochev.
{\em Spectral flow is the integral of a one form on the {B}anach manifold of self-adjoint {F}redholm operators},
Adv. Math., {\bf 222} (5), 2009, 1809--1849.

\bibitem{colin81}
Y. Colin de Verdi\`ere.
{\em Une formule de traces pour l'op\'{e}rateur de {S}chr\"{o}dinger dans {$\R^3$}}, 
Ann. Sci. \'{E}cole Norm. Sup., {\bf 14} (1), 1981, 27--30.


%\bibitem{zworski19}
%S. Dyatlov, M. Zworski. {\em Mathematical theory of scattering resonances}, volume 200 of
%  Graduate Studies in Mathematics. American Mathematical Society, Providence, RI, 2019.

%\bibitem{erdogan14}
%M. B. Erdo\u{g}an, M. Goldberg, W. R. Green. {\em Dispersive estimates for four dimensional {S}chr\"{o}dinger and wave equations with obstructions at zero energy}, Comm. Partial Differential Equations, {\bf 39} (10), 2014, 1936--1964.

\bibitem{folland}
G. B. Folland.
{\em How to integrate a polynomial over a sphere}, Amer. Math. Monthly, {\bf 108} (5), 2001.

%\bibitem{guillope81}
%L. Guillop\'{e}. {\em Asymptotique de la phase de diffusion pour l'op\'{e}rateur de {S}chr\"{o}dinger avec potentiel}, C. R. Acad. Sci. Paris S\'{e}r. I Math., {\bf 293} (12), 1981.


\bibitem{guillope81}
L. Guillop\'{e}. {\em Une formule de trace pour l'op\'{e}rateur de {S}chr\"{o}dinger}, PhD Thesis, Universit\'{e} Joseph Fourier Grenoble, 1981. Available at \href{https://www.math.sciences.univ-nantes.fr/~guillope/LG/these_1981.pdf}{https://www.math.sciences.univ-nantes.fr/$\sim$guillope/LG/these\textunderscore1981.pdf}

%\bibitem{HR} N. Higson, J. Roe, {\em Analytic $K$-Homology},
%Oxford University Press, 2000.
%
%\bibitem{inoue19}
%H. Inoue, N. Tsuzu. {\em Schr\"{o}dinger wave operators on the discrete half-line}, Integral Equations Operator Theory, {\bf 91} (5), 2019.
%
%\bibitem{inoue19ii}
%H. Inoue, S. Richard. {\em Topological {L}evinson's theorem for inverse square potentials: complex, infinite, but not exceptional}, Rev. Roumaine Math. Pures Appl., {\bf 64} (2-3), 2019, 225-250.
%
%\bibitem{inoue20}
%H. Inoue. {\em Explicit formula for {S}chr\"{o}dinger wave operators on the half-line for potentials up to optimal decay}, J. Funct. Anal., {\bf 279} (7), 2020, 108630.
%
%\bibitem{isozaki12}
%H. Isozaki, S. Richard. {\em On the wave operators for the {F}riedrichs-{F}addeev model}, Ann. Henri Poincar\'{e}, {\bf 1} (6), 2012, 1469-1482.

%\bibitem{holub72}
%J. R. Holub. {\em Compactness in topological tensor products and operator spaces}, Proc. Amer. Math. %Soc., {\bf 36}, 1973, 398--406.
%
%\bibitem{jeffrey08}
%A. Jeffrey, H-H. Dai. {\em Handbook of mathematical formulas and integrals}, Elsevier/Academic Press, Amsterdam, fourth edition, 2008.

\bibitem{jensen79}
A. Jensen, T. Kato. {\em Spectral properties of {S}chr\"{o}dinger operators and time-decay of
  the wave functions}, Duke Math. J., {\bf 46} (3), 1979, 583--611.

\bibitem{jensen80}
A. Jensen. {\em Spectral properties of {S}chr\"{o}dinger operators and time-decay of the wave functions results in {$L^{2}(\R^{m})$}, {$m\geq 5$}},  Duke Math. J., {\bf 47} (1), 1980, 57--80.

\bibitem{jensen81}
A. Jensen. {\em Time-delay in potential scattering theory. {S}ome ``geometric''
  results}, Comm. Math. Phys., {\bf 82} (3), 1981, 435--456.

%\bibitem{jensen84}
%A. Jensen. {\em Spectral properties of {S}chr{\"{o}}dinger operators and time-decay
%  of the wave functions. {R}esults in {$L^{2}(\R^{4})$}},  J. Math. Anal. Appl., {\bf 101} (2), 1984, 397--422.
%
%\bibitem{jensen01}
%A. Jensen, G. Nenciu. {\em A unified approach to resolvent expansions at thresholds}, Rev. Math. Phys, {\bf 13} (6), 2001, 717--754.

\bibitem{jia12}
X. Jia, F. Nicoleau, X.P. Wang.
{\em A new {L}evinson's theorem for potentials with critical decay},
Ann. Henri Poincar\'{e}, {\bf 13} (1), 2012, 41-84.

%\bibitem{kato66}
%T. Kato. {\em Wave operators and similarity for some non-selfadjoint operators}, Math. Ann., {\bf 162}, 1966, 258--279.

%\bibitem{kellendonk11}
%J. Kellendonk, K. Pankrashkin, S. Richard. {\em Levinson's theorem and higher degree traces for {A}haronov-{B}ohm operators}, J. Math. Phys., {\bf 52} (5), 2011, 052102.
 
\bibitem{kellendonk06}
J. Kellendonk, S. Richard.
{\em Levinson's theorem for {S}chr{\"{o}}dinger operators with point interaction: a topological approach},  J. Phys. A, {\bf 39} (46), 2006, 14397--14403.

\bibitem{kellendonk08}
J. Kellendonk, S. Richard.
{\em  On the structure of the wave operators in one dimensional potential
  scattering}, Math. Phys. Electron. J., {\bf 14}, 1-21, 2008.


\bibitem{kellendonk12}
J. Kellendonk, S. Richard. {\em On the wave operators and {L}evinson's theorem for potential
  scattering in {$\mathbb{R}^3$}}, Asian-Eur. J. Math., {\bf 5} (1), 2012.

\bibitem{krein53}
M. G. Kre\u{\i}n.
{\em On the trace formula in perturbation theory}, 
Mat. Sbornik N.S., {\bf 33 \slash 75}, 1953, 597--626.

%\bibitem{KS80} 
%M. Klaus, B. Simon.
%{\em Coupling constant thresholds in nonrelativistic quantum mechanics. {I}. Short-range two-body case}, Ann. Physics, {\bf 130} (2), 1980, 251--281.
  
%\bibitem{kuroda78}
%S. T. Kuroda. \newblock {\em An introduction to scattering theory}, volume 51 of {\em Lecture
%  Notes Series}. Aarhus Universitet, Matematisk Institut, Aarhus, 1978.

\bibitem{lesch05}
M. Lesch.
{\em The uniqueness of the spectral flow on spaces of unbounded self-adjoint {F}redholm operators}, 
in `Spectral geometry of manifolds with boundary and decomposition of manifolds', ser. Contemp. Math. {\bf 366}, Amer. Math. Soc., 2005, 193--224.

\bibitem{levinson49}
N. Levinson. {\em On the uniqueness of the potential in a {S}chr\"{o}dinger equation for a given asymptotic phase}, Danske Vid. Selsk. Mat.-Fys. Medd., {\bf 25} (9), 1949.
  
%  \bibitem{macaulay39} 
%  P. Macaulay-Owen. {\em Parseval's Theorem for Hankel Transforms}, Proc. Lond. Math. Soc., {\bf 45}, 1939, 458--474.
  
%\bibitem{melgaard01}
%M. Melgaard. {\em Spectral properties in the low-energy limit of one-dimensional {S}chr\"{o}dinger operators $H = -d^2\slash dx^2+V$. The case $\langle 1, V1 \rangle \neq 0$}, Math. Nachr., {\bf 238}, 2002, 113--143.
%
%\bibitem{murata84}
%M. Murata. {\em High energy resolvent estimates. {II}. {H}igher order elliptic operators}, J. Math. Soc. Japan, {\bf 36} (1), 1984, 1--10.

\bibitem{newton77}
R. G. Newton.
{\em Noncentral potentials: the generalized Levinson theorem and the structure of the spectrum}, J. Mathematical Phys., {\bf 18} (7), 1977, 1348-1357.

%\bibitem{oberhettinger74} F. Oberhettinger. {\em Tables of Mellin transforms}, Springer-Verlag, New York-Heidelberg, 1974.

%\bibitem{richard14}
%K. Pankrashkin, S. Richard. {\em One-dimensional {D}irac operators with zero-range interactions: spectral, scattering, and topological results}, J. Math. Phys. {\bf 55} (6), 2014, 062305.
%
%\bibitem{pearson88}
%D. B. Pearson. {\em Quantum scattering and spectral theory}, volume 9 of {\em Techniques of Physics}. Academic Press, Inc., Harcourt Brace Jovanovich, Publishers, London, 1988.

\bibitem{phillips96}
J. Phillips.
{\em Self-adjoint {F}redholm operators and spectral flow},
Canad. Math. Bull. {\bf 39} (4), 1996, 460--467.

\bibitem{phillips97}
J. Phillips.
{\em Spectral flow in type {I} and {II} factors --- a new approach}, 
in `Cyclic cohomology and noncommutative geometry', Fields Inst. Commun., {\bf 17}, 1997, 137--153.

%\bibitem{popov82}
%G. S. Popov.
%{\em Asymptotic behaviour of the scattering phase for the {S}chr\"{o}dinger operator},
%C. R. Acad. Bulgare Sci. {\bf 25} (7), 1982, 885-888.

%\bibitem{prolla77}
%J. B. Prolla. {\em Approximation of vector valued functions}, North-Holland Mathematics Studies, Vol. 25. North-Holland Publishing Co., Amsterdam-New York-Oxford, 1977.

\bibitem{reedsimon79}
M. Reed, B. Simon.
{\em Methods of modern mathematical physics {III}: Scattering theory},
Academic Press, New York-London, 1979.

%\bibitem{richard10}
%S. Richard, R. Tiedra de Aldecoa. {\em New formulae for the wave operators for a rank one interaction}, Integral Equations and Operator Theory, {\bf 66} (2), 2010, 283--292.

\bibitem{richard13}
S. Richard, R. Tiedra de Aldecoa. {\em New expressions for the wave operators of {S}chr\"{o}dinger operators in {$\mathbb{R}^3$}}, Lett. Math. Phys., {\bf 103} (11), 2013, 1207--1221.

\bibitem{richard13ii} S. Richard, R. Tiedra de Aldecoa. {\em Explicit formulas for the {S}chr\"{o}dinger wave operators in {$\mathbb{R}^2$}}, C. R. Math. Acad. Sci. Paris, {\bf 351} (5-6), 2013, 209--214.

%\bibitem{richard20} S. Richard, T. Umeda. {\em On some integral operators appearing in scattering theory, and their resolutions}, in ``Spectral theory and mathematical physics'', Springer, Cham, 2020, 243--256.

\bibitem{richard21}
S. Richard, R. Tiedra de Aldecoa, L. Zhang.
{\em Scattering operator and wave operators for 2{D} {S}chr\"{o}dinger operators with threshold obstructions}, Complex Anal. Oper. Theory, {\bf 15} (6) 2021.

%\bibitem{simon79}
%B. Simon. {\em Trace Ideals and Their Applications}, Mathematical surveys and monographs. American Mathematical Society, 2005.

%\bibitem{yafaev92}
%D. R. Yafaev. {\em Mathematical Scattering Theory: General Theory}, {\bf 105}, Translations of mathematical monographs, American Mathematical Society, 1992.

\bibitem{yafaev10}
D. R. Yafaev. {\em Mathematical scattering theory: Analytic theory}, {\bf 158}, 
  Mathematical surveys and monographs, American Mathematical Society, 2010.





\end{thebibliography}
\end{document}